\documentclass[journal,10pt]{IEEEtran}
\usepackage{amsmath}

\usepackage{graphicx}
\usepackage{array}
\usepackage{mdwmath}
\usepackage{amssymb}
\usepackage{mdwtab}
\usepackage{amsmath,amsthm}
\usepackage{threeparttable}
\usepackage{color}
\usepackage{bm}
\usepackage{url}
\usepackage{algpseudocode}
\usepackage{algorithm}
\usepackage{multicol}
\usepackage{epstopdf}
\usepackage[colorlinks]{hyperref}

\usepackage{flushend}
\usepackage[caption=false,font=footnotesize]{subfig}

\algrenewcommand{\algorithmicrequire}{\textbf{Input:}}
\algrenewcommand{\algorithmicensure}{\textbf{Output:}}

\newtheorem{theorem}{Theorem}
\newtheorem{remark}{Remark}
\newtheorem{lemma}{Lemma}

\begin{document}

\title{
%Fluid Antenna System Enabled UAV Communications for URLLC
Fluid Antenna System-Enabled UAV Communications in the Finite Blocklength Regime
%Fluid Antenna System Enabled UAV Communications in Finite Blocklength Regime
%Fluid Antenna Enabled UAV Cooperative Communications: BLER Analysis and Optimization
%Fluid Antenna Enabled UAV Cooperative Communications: Finite Blocklength Analysis and Energy Efficiency Optimization
%Exploiting FAS in Low-Altitude UAV Communications: BLER  Analysis and Energy Efficiency Optimization
}

\author{Xusheng Zhu,
            Kai-Kit Wong, \IEEEmembership{Fellow,~IEEE},
            Hanjiang Hong,
            Han Xiao,\\
            Hao Xu,
            Tuo Wu, and
            Chan-Byoung Chae~\IEEEmembership{Fellow,~IEEE}
\vspace{-7mm}

\thanks{(\emph{Corresponding author: Kai-Kit Wong}).}
\thanks{The work of X. Zhu and K. K. Wong was in part supported by the Engineering and Physical Sciences Research Council (EPSRC) under grant EP/W026813/1.}
\thanks{The work of H. Hong was supported by the Outstanding Doctoral Graduates Development Scholarship of Shanghai Jiao Tong University.}
\thanks{The work of C.-B. Chae was in part supported by the Institute for Information and Communication Technology Planning and Evaluation (IITP)/NRF grant funded by the Ministry of Science and ICT (MSIT), South Korea, under Grant RS-2024-00428780 and 2022R1A5A1027646.}

\thanks{X. Zhu, K. K. Wong, and H. Hong are with the Department of Electronic and Electrical Engineering, University College London, London, United Kingdom. (e-mail: \{xusheng.zhu;kai-kit.wong;hanjiang.hong\}@ucl.ac.uk). K. K. Wong is also affiliated with the Yonsei Frontier Lab., Yonsei University, Seoul 03722, Korea.}
\thanks{H. Xiao is School of Information and Communications Engineering, Xi'an Jiao Tong University, China (e-mail: hanxiaonuli@stu.xjtu.edu.cn).}
\thanks{H. Xu is with the National Mobile Communications Research Laboratory, Southeast University, Nanjing 210096, China (e-mail: hao.xu@seu.edu.cn).}
\thanks{T. Wu is with the City University of Hong Kong, Hong Kong (e-mail: tuowu2@cityu.edu.hk).}
\thanks{C.-B. Chae is with the School of Integrated Technology, Yonsei University, Seoul 03722, Republic of Korea (e-mail: cbchae@yonsei.ac.kr).}
}

\maketitle

\begin{abstract}
This paper develops a comprehensive framework for the performance analysis of fluid antenna system (FAS)-enabled unmanned aerial vehicle (UAV) relaying networks operating in the finite blocklength regime. Our contribution lies in establishing a rigorous methodology for characterizing system reliability under diverse propagation environments. Closed-form expressions for the block error rate (BLER) are derived by employing a tractable eigenvalue-based approximation of the spatially correlated UAV-to-user link, whose underlying independent diversity components are modeled as Nakagami-$m$ fading. This approach addresses both line-of-sight (LoS) dominant rural and probabilistic non-line-of-sight (NLoS) urban scenarios. Furthermore, a high signal-to-noise ratio (SNR) asymptotic analysis is developed, revealing the fundamental diversity order of the UAV-to-user link. Based on this, we further address the practical issue of energy efficiency. A realistic energy efficiency maximization problem is formulated, which explicitly accounts for the time and energy overhead inherent in the FAS port selection process, a factor often omitted in idealized models. An efficient hierarchical algorithm is then proposed to jointly optimize the key system parameters. Extensive numerical results validate the analysis and illustrate that while FASs can yield substantial power gains, the operational overhead introduces a non-trivial trade-off. This trade-off leads to an optimal number of ports and fundamentally different UAV deployment strategies in rural versus urban environments. This work provides both foundational analysis and practical design guidelines for FAS-enabled UAV communications.
\end{abstract}


\section{Introduction}
\IEEEPARstart{T}{he} evolution toward sixth-generation (6G) wireless networks is driven by the imperative to support a new wave of mission-critical services, including autonomous systems \cite{rahim-2022} and tactile internet \cite{Simsek-2016}, which are fundamentally underpinned by ultra-reliable low-latency communications (URLLC)~\cite{Popovski-2019}. A cornerstone for URLLC is the need of short-packet communication to meet stringent latency requirements, often below one millisecond~\cite{she2017radio}. Nevertheless, operating within this finite blocklength regime fundamentally alters the performance landscape, rendering classical Shannon capacity theorems, predicated on infinitely long codewords, an inadequate benchmark \cite{Durisi-2016}. Instead, system performance in URLLC is governed by information-theoretic limits that explicitly articulate the trade-off between coding rate, blocklength, and a non-negligible decoding error probability~\cite{polyanskiy2010channel}.

On the other hand, unmanned aerial vehicles (UAVs) have emerged as a pivotal technology for enhancing the flexibility, reliability, and coverage of 6G networks~\cite{wu2021overview,xiao2025ene,xiao2025star}. Deployed as aerial relays, UAVs can dynamically establish robust communication links, circumventing blockages and providing seamless connectivity where terrestrial infrastructure is compromised \cite{ren2020joint,yuan2022per,elw2023power,Liu-2024}. While the synergy between UAVs and URLLC has been an active area of research, a key challenge persists in providing reliable links to ground users, such as mobile handsets and Internet of Things (IoT) devices, which are subject to strict size, weight, and power (SWaP) constraints. In this case, achieving spatial diversity with conventional multiple-antenna technology is often impractical.

To tackle this, the fluid antenna system (FAS) concept has emerged as a transformative paradigm capable of harvesting spatial diversity within a compact form factor \cite{new2025a,Lu-2025,wu2024flu,New-2026jsac}. First proposed by Wong {\em et al.}~in \cite{wong2021fluid,wong2020pel}, FAS treats antenna as a reconfigurable physical-layer resource to broaden system design and network optimization. Unlike fixed-position antennas (FPAs), one feature of FAS is to be able to optimize antenna position with fine resolution for extraordinary spatial diversity \cite{new2024fga,New2024information,saixu2025fas}. Channel estimation for FAS systems has been addressed in \cite{xu2024channel,zhang2023successive}. In \cite{new2025channe}, the trade-off between system performance and channel estimation overhead in FAS was further analyzed. FAS also provides a new way to mitigate interference, leading to the concept of fluid antenna multiple access (FAMA) \cite{wong2022fama,wong2023ffast,wong2023sfast,hong2025down,wong2024compact,Waqar-2025}.

The potential of FAS to enhance link reliability and spectral efficiency by leveraging its unique spatial selectivity and hardware-efficient single-radio frequency (RF)-chain architecture, making it an attractive candidate for SWaP-constrained platforms \cite{zhu2025fas}. As such, a burgeoning body of research has begun to explore the integration of FAS with UAVs. For instance, the authors in \cite{cha2025uav} investigated a multiuser network with UAV as a relay, where the base station (BS) employs rate-splitting multiple access. In \cite{li2025radi}, the authors investigated interference resilience through radiation pattern control. Other pioneering works have explored FAS-assisted 3D UAV positioning \cite{xu2025fluid}, the synergy with reconfigurable intelligent surfaces (RIS) in UAV networks \cite{shen2025ris}, and the crucial task of dynamic channel modeling for FAS-UAV communications \cite{jiang2025dyn}.

Despite these promising advancements, the application of FAS to URLLC-enabled UAV missions remains nascent. A critical limitation of these prior works is that their analyses are implicitly rooted in the infinite blocklength regime, relying on classical metrics like ergodic capacity and outage probability. Such metrics, predicated on capacity-achieving infinitely long codewords, are inadequate for URLLC, which relies on short-packet communication to meet stringent latency requirements. Operating in the finite blocklength regime introduces a fundamentally different performance paradigm, governed by the non-trivial trade-off between reliability and latency.

This finite blocklength challenge, coupled with other existing gaps, motivates this work. First, current research lacks a unified analytical framework that jointly considers the compact diversity of FAS-UAVs under these structures of finite blocklength theory. Second, existing studies often rely on oversimplified channel models, failing to capture the diverse nature of air-to-ground propagation across different environments. Thirdly, a universal oversight is the neglect of the overhead associated with practical FAS port selection, which incurs non-trivial time delay and energy consumption. Last but not least, the analytical intractability imposed by the complex statistics of spatially correlated channels further compounds these challenges, hindering the development of tractable performance evaluation and optimization frameworks.

Against this background, we aim to bridge these significant gaps by developing a holistic framework for the analysis and optimization of an FAS-enabled UAV relaying network. Our main contributions are summarized as follows:
\begin{enumerate}
\item We establish a comprehensive analytical framework for evaluating the performance of FAS-UAV systems in the finite blocklength regime. To overcome the mathematical intractability of spatially correlated channels, we employ a tractable model based on the channel's underlying eigenvalue-weighted independent diversity branches. This approach facilitates the derivation of closed-form expressions for the block error rate (BLER) considering Nakagami-$m$ fading. The framework is sufficiently general to characterize performance across fundamentally different propagation environments, providing distinct analyses for both line-of-sight (LoS)-dominant rural and probabilistic non-LoS (NLoS) urban scenarios.

\item We derive insightful asymptotic expressions for the BLER that characterize the system's fundamental performance limits in the high signal-to-noise ratio (SNR) regime. This analysis explicitly quantifies the system's diversity order, revealing that reliability scales with the product of the channel's fading severity parameter and the effective spatial degrees of freedom harvested by the FAS. Furthermore, our analysis proves the existence of an error floor in both scenarios, a critical insight demonstrating that the overall system performance becomes fundamentally bottlenecked by the reliability of the first-hop (BS-to-UAV) link at high UAV transmit powers.

\item We formulate an energy efficiency (EE) maximization problem that holistically incorporates the non-negligible time and energy overheads inherent to the FAS port selection mechanism. Unlike idealized models prevalent in the literature, this realistic formulation captures the critical trade-off between the diversity gains afforded by an increasing number of FAS ports and the corresponding operational costs. To address this non-convex mixed-integer problem, we propose an efficient hierarchical algorithm that decouples the optimization variables to find a high-quality solution with low complexity.

\item Through simulations, we validate our analytical framework and uncover several critical design insights for practical FAS-UAV deployment. Our results demonstrate the existence of an optimal, finite number of FAS ports that maximizes EE, which is a direct consequence of the trade-off between diminishing diversity returns and increasing operational overhead. We also reveal that the optimal UAV deployment strategy is fundamentally different in rural versus urban environments. Additionally, the optimal altitude in urban settings is governed by a non-trivial balance between path loss and blockage probability, in stark contrast to the monotonic behavior observed in rural scenarios. Ultimately, this comprehensive framework facilitates the design and deployment of ultra-reliable low-latency UAV communications, essential for mission-critical IoT applications and the realization of robust aerial edge networks.
\end{enumerate}

The rest of this paper is organized as follows. Section~\ref{sec:system_model} presents the system and channel models. Section~\ref{sec:analysis} provides the detailed performance analysis. Section~\ref{sec:eemax} formulates and solves the EE maximization problem. In Section~\ref{sec:results}, we present numerical results, and Section~\ref{sec:conclude} concludes the paper.

 \emph{Notations:}
 Scalars are denoted by italic letters, while vectors and matrices are denoted by boldface lowercase and uppercase letters, respectively.
 Sets are represented by calligraphic letters (e.g., $\mathcal{K}$).
 $\mathbb{C}^{M \times N}$ denotes the space of $M \times N$ complex-valued matrices.
 $J_0(\cdot)$ is a first-order modified Bessel function of the zeroth order.
 $\text{rank}(\mathbf{X})$ denote its rank. $\text{diag}(x_1, \dots, x_N)$ returns a diagonal matrix with diagonal elements $x_1, \dots, x_N$.
 $\mathcal{CN}(\mu, \sigma^2)$ denotes a circularly symmetric complex Gaussian (CSCG) distribution with mean $\mu$ and variance $\sigma^2$.
 %The symbol $\sim$ is used to denote ``is distributed as". The symbol $\in$ denotes ``is an element of".

\begin{figure*}[t]
\centering
\subfloat[Rural scenario]{%
\includegraphics[width=0.48\textwidth]{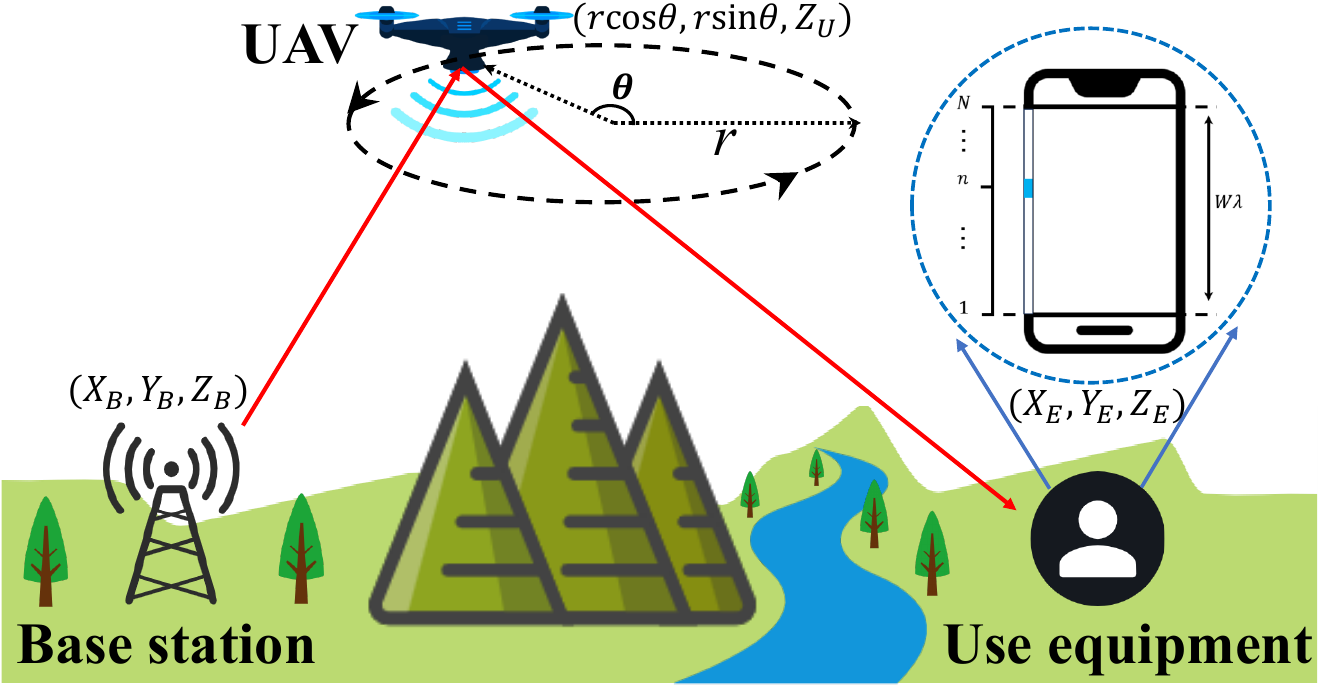}\label{fig:rural}}
\hfil
\subfloat[Urban scenario]{%
\includegraphics[width=0.48\textwidth]{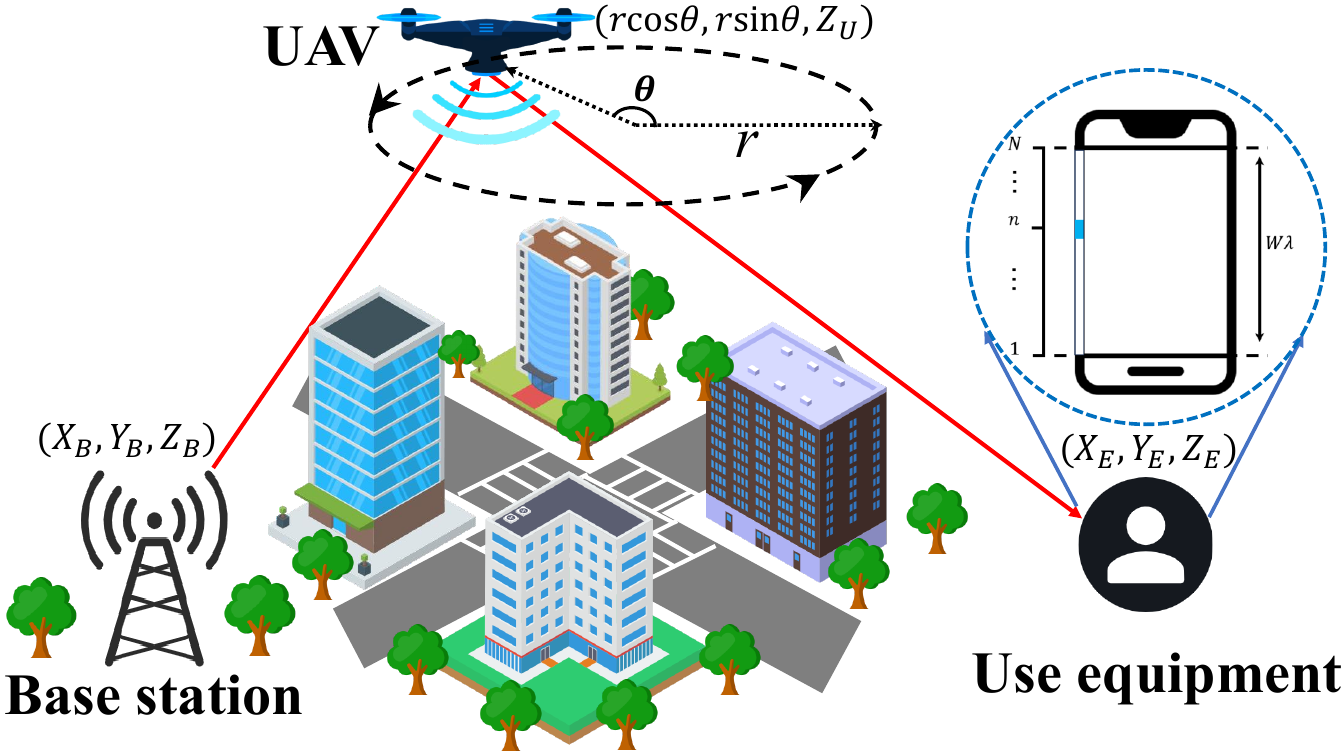}\label{fig:urban}}
\caption{A FAS-enabled UAV cooperative short-packet system under the rural and urban scenarios.}\label{fig:frame}
\vspace{-3mm}
\end{figure*}

\section{The FAS-UAV System Model}\label{sec:system_model}
We consider a three-node, UAV-aided downlink communication system, as illustrated in Fig.~\ref{fig:frame}. The system comprises a BS that transmits data to a user equipment (UE) via a half-duplex, decode-and-forward (DF) UAV relay. A key feature of this architecture is the adoption of FAS at the UE to provide spatial diversity and enhance link reliability. To facilitate a comprehensive analysis, the subsequent sections present the detailed models for the wireless channels, signal structure, and path loss characteristics pertinent to this system architecture under both rural and urban deployment scenarios.

\subsection{Channel Model}\label{ssec:model}
We assume the $N$ ports of FAS are uniformly distributed over a linear aperture of length $W\lambda$,\footnote{The results in this paper can be extended to cope with multi-dimensional FAS structures while this paper opts for the linear structure for simplicity.} inducing spatial correlation across the small-scale fading channels. We characterize this correlation using the Hermitian Toeplitz matrix $\mathbf{J}$, whose elements are given by the Jakes' model, $J_{m,n} = J_0( 2\pi W\,(m-n)/(N-1) )$ \cite{New2024information,xzhu2023fris}. The eigendecomposition of the correlation matrix is denoted by $\mathbf{J}=\mathbf{U}\mathbf{\Lambda}\mathbf{U}^{H}$, where $\mathbf{U}$ is a unitary matrix and $\mathbf{\Lambda}=\mathrm{diag}(\lambda_1,\dots,\lambda_N)$ contains the eigenvalues in non-increasing order. This allows the correlated channel vector $\mathbf{h}=[h_1,\dots,h_N]^T$ to be represented as
\begin{equation}\label{mahug}
\mathbf{h} = \mathbf{U}\,\mathbf{\Lambda}^{1/2}\mathbf{g},
\end{equation}
in which $\mathbf{g}=[g_1,\dots,g_N]^T$ is a vector of independent and identically distributed (i.i.d.) baseband components. To capture a wide range of fading conditions, we adopt the Nakagami-$m$ distribution, where the power of each component, $|g_n|^2$, follows a Gamma distribution with severity parameter $m_2$.

Crucially, as a linear combination of Nakagami-$m$ variates, the physical channel $h_n$ at each port is not Nakagami-$m$ distributed, and its statistics is mathematically intractable. A direct analysis of the true selected channel gain, $\max\{|h_n|^2\}$, is therefore infeasible. To circumvent this, we adopt a widely used and effective approach based on the channel's underlying independent diversity branches \cite{zhao2025analy}. This method provides a tractable yet insightful framework that captures the essence of the diversity offered by the FAS. The number of effective diversity branches is given by the rank of the correlation matrix, $N_{\text{eff}} \triangleq \text{rank}(\mathbf{J})$. Accordingly, the analytical model for the selected channel power gain is defined as \cite{zhu2025fas,zhao2025analy,zhu2025ua}
\begin{equation}\label{eq:analytical_model_def}
|h_{\rm FAS}|^2 \triangleq \max\{\lambda_1|g_1|^2, \dots, \lambda_{N_{\mathrm{eff}}}|g_{N_{\mathrm{eff}}}|^2\}.
\end{equation}

\subsection{Relaying Protocol and Signal Model}
The communication protocol follows a two-phase DF relaying scheme. The received baseband signals at the UAV and UE are, respectively, expressed as
\begin{align}
y_1 &= \sqrt{P_1\beta_1}gx_t + n_1, \label{eq:y1_signal} \\
y_2  &= \sqrt{P_2\beta_2} h_{\rm FAS}x_t + n_2, \label{eq:y2_signal}
\end{align}
where $g$ is the first-hop channel coefficient, with its power $|g|^2$ following a Gamma distribution with parameter $m_1$. The term $h_{\rm FAS}$ is the effective second-hop channel coefficient, whose power gain is given by the model in \eqref{eq:analytical_model_def}. Other parameters include the transmit powers $P_1, P_2$; large-scale path loss coefficients $\beta_1, \beta_2$; transmit signal $x_t$; and additive white Gaussian noise (AWGN) terms $n_1, n_2 \sim\mathcal{CN}(0,\sigma^2)$.

\subsection{Path Loss and Large-Scale Fading}
We consider a three-dimensional (3-D) Cartesian coordinate system where the UAV flies on a circular trajectory of radius $r$ at altitude $Z_{\rm U}$, while the BS and UE are at fixed locations. The slant ranges for the BS-UAV and UAV-UE links, $d_1(\theta)$ and $d_2(\theta)$, are, respectively, given by
\begin{equation}
\left\{\begin{aligned}
d_1(\theta)&=\sqrt{(r\cos\theta\!-\!X_{\rm B})^2 \!+\! (r\sin\theta\!-\!Y_{\rm B})^2 \!+\! (Z_{\rm U}\!-\!Z_{\rm B})^2},\\
d_2(\theta)&=\sqrt{(r\cos\theta\!-\!X_{\rm E})^2 \!+\! (r\sin\theta\!-\!Y_{\rm E})^2 \!+\! (Z_{\rm U}\!-\!Z_{\rm E})^2}.
\end{aligned}\right.
\end{equation}

\subsubsection{Rural Scenario}
The rural environment, see Fig.~\ref{fig:frame}(a), is characterized by predominantly LoS links, for which the free-space path loss model applies. The path loss in dB is given by $L_{i}^{\rm FS} = 20\log_{10}(4\pi f_c d_{i}(\theta)/c)$, $i\in\{1,2\}$, where $f_c$ is the carrier frequency and $c$ is the speed of light. Accordingly, the large-scale fading coefficient is $\beta_{i}^{\rm RS} = (c/4\pi f_c d_i(\theta))^2$. This path loss determines the average SNRs for the two hops as
\begin{align}
\bar{\gamma}_1^{\rm RS}(\theta) &= \frac{P_1 \beta_1^{\rm RS}(\theta)}{\sigma^2}, \\
\bar{\gamma}_2^{\rm RS}(\theta) &= \frac{P_2 \beta_2^{\rm RS}(\theta) \sum_{n=1}^{N_{\rm eff}} \lambda_n}{\sigma^2},\label{eq:gamma2}
\end{align}
where the normalization factor $\sum_{n=1}^{N_{\rm eff}}\lambda_{n} = \mathbb{E}[\sum_{n=1}^{N_{\rm eff}}\lambda_n|g_n|^2]$ is the total average power gain across all decorrelated diversity branches. This term serves as an analytical convenience and provides an upper bound on the true average selected channel gain, $\mathbb{E}[\max_{n}\{\lambda_n|g_n|^2\}]$. The instantaneous SNR for each hop is then obtained by incorporating the small-scale fading as
\begin{align}
\gamma_1^{\rm RS} &= \bar{\gamma}_1^{\rm RS}(\theta) |g|^2, \label{eq:snr1_rural_inst} \\
\gamma_{2}^{\rm RS}(\theta) &= \overline{\gamma}_{2}^{\rm RS}(\theta) \frac{|h_{\rm FAS}|^{2}}{\sum_{n=1}^{N_{\rm eff}}\lambda_{n}}. \label{gamsimax}
\end{align}

\subsubsection{Urban Scenario}
In contrast to the rural setting, the urban scenario depicted in Fig.~\ref{fig:frame}(b) incorporates a probabilistic mixture of LoS and NLoS conditions. The path loss model is augmented with an excess loss term $\eta^k$ for each link type $k \in \{\rm LoS, NLoS\}$, such that $L_{i}^{k} = 20\log_{10}(4\pi f_c d_{i}(\theta)/c) + \eta_{i}^{k}$. Consequently, the large-scale fading coefficient is $\beta_{i}^{k} = (c/4\pi f_c d_i(\theta))^2  10^{-\eta_{i}^{k}/10}$. For a given hop type $k$, the average SNRs are defined as
\begin{align}
\bar{\gamma}_1^k(\theta) &= \frac{P_1 \beta_1^k(\theta)}{\sigma^2}, \\
\bar{\gamma}_2^k(\theta) &= \frac{P_2 \beta_2^k(\theta) \sum_{n=1}^{N_{\rm eff}} \lambda_n}{\sigma^2}.
\end{align}
The corresponding instantaneous SNRs are given by
\begin{align}
\gamma_1^k &= \bar{\gamma}_1^k(\theta) |g^k|^2, \\
\gamma_{2}^{k}(\theta) &= \overline{\gamma}_{2}^{k}(\theta) \frac{|h_{\rm FAS}^{k}|^{2}}{\sum_{n=1}^{N_{\rm eff}}\lambda_{n}}.
\end{align}
The probability of an LoS condition for hop $i$ is modeled as a function of the elevation angle $\phi_i(\theta)$ as \cite{hour2014opt}
\begin{equation}
P^{\rm LoS}_{i}(\theta)=\frac{1}{1+a\exp\left(-b[\phi_i(\theta)-a]\right)},
\end{equation}
with $P^{\rm NLoS}_{i}(\theta)=1-P^{\rm LoS}_{i}(\theta)$. The elevation angle in degrees, $\phi_i(\theta)$, is given by $\phi_1(\theta)=\frac{180}{\pi}\arcsin((Z_{\rm U}-Z_{\rm B})/d_1(\theta))$ and $\phi_2(\theta)=\frac{180}{\pi}\arcsin((Z_{\rm U}-Z_{\rm E})/d_2(\theta))$. For a dense urban environment, empirically, we have $a=12.08$ and $b=0.11$.

\section{Performance Analysis}\label{sec:analysis}
For a given finite blocklength and a fixed target error probability of the system, the coding rate $R(l,\epsilon)$ of a finite blocklength packet is approximated as \cite{polyanskiy2010channel}
\begin{equation}\label{Rnx}
R(L,\epsilon)=\frac{B}{L}
= C(\gamma)-\sqrt{\frac{Z(\gamma)}{L}}Q^{-1}(\epsilon)+O\left(\frac{\log_2L}{L}\right),
\end{equation}
where $B$ is the packet size in terms of the number of transmission bits, $L$ is the blocklength, $C(\gamma)=\log_2(1+\gamma)$ is the Shannon capacity, $Z(\gamma)=\left(1-\frac{1}{(1+\gamma)^2}\right)(\log_2e)^2$, measured in squared information units per channel usage, refers to the channel dispersion that represents the variation of stochastic channel compared with a deterministic channel for the same capacity, $\epsilon$ is the expected error probability, $Q^{-1}(\cdot)$ is the inverse Gaussian Q-function $Q(x)=\frac{1}{2\pi}\int_x^\infty \exp(-\frac{t^2}{2})dt$, and the term $O\left(\frac{\log_2L}{L}\right)$ represents a remainder of order $\frac{\log_2L}{L}$, which becomes negligible for $L\geq 100$.

\subsection{Rural Scenario}
According to \eqref{Rnx}, the instantaneous BLER in the $i$-th phase is formally expressed as
\begin{equation}
\epsilon_i^{\rm RS}\approx Q\left(\frac{C(\gamma^{\rm RS}_i)-R}{\sqrt{Z(\gamma^{\rm RS}_i)/L}}\right).
\end{equation}
In this manner, the average BLER at the UE under the given $\theta$, which can be calculated as
\begin{equation}\label{baepsn}
\bar{\epsilon}_i^{\rm RS}\approx\int_0^\infty Q\left(\frac{C(\gamma^{\rm RS}_i)-R}{\sqrt{Z(\gamma^{\rm RS}_i)/L}}\right)f_{\gamma_i^{\rm RS}}(x)dx,
\end{equation}
where $f_{\gamma_i^{\rm RS}}(x)$ denotes the probability density function (PDF) of ${\gamma_i^{\rm RS}}$. However, a direct evaluation of the integral in \eqref{baepsn} is analytically intractable due to the complex form of the Q-function in the integrand. To render this integral tractable, we employ an accurate piecewise linear approximation for the Q-function within a defined SNR range $[\rho_L, \rho_H]$, given by \cite{yuan2022per}
\begin{equation}
Q\left(\frac{C(\gamma^{\rm RS}_i)-R}{\sqrt{Z(\gamma^{\rm RS}_i)/L}}\right)=
\begin{cases}
1, & \gamma^{\rm RS}_i \leq \rho_L,\\
\frac{1}{2} - \chi(\gamma^{\rm RS}_i - \tau), & \rho_L < \gamma^{\rm RS}_i < \rho_H,\\
0, & \gamma^{\rm RS}_i \geq \rho_H,
\end{cases}
\end{equation}
where $\chi = 1/\sqrt{2\pi(2^{R}-1)/L}$, $\tau=2^{R}-1$, $\rho_L = \tau-1/(2\chi)$, and $\rho_H=\tau+1/(2\chi)$. Thus, the average BLER becomes
\begin{equation}\label{epaprir}
\bar{\epsilon}^{\rm RS}_i\approx\chi\int_{\rho_L}^{\rho_H} F_{\gamma_i^{\rm RS}}(x)dx,
\end{equation}
in which $F_{\gamma_i^{\rm RS}}(x)$ denotes the cumulative distribution function (CDF) of $\gamma_i^{\rm RS}$. Afterwards, the average BLER for each hop is derived separately. We begin with the first hop.

\subsubsection{First-Hop Average BLER}
The average BLER for the first hop, $\overline{\epsilon}_{1}^{\mathrm{RS}}$, is determined by the statistics of $\gamma_1^{\mathrm{RS}}$. Its CDF is provided in the following lemma.

\begin{lemma}\label{lema1}
The expression for ${{\rm F}_{\gamma_1^{\rm RS}}(x)}$ can be given by
\begin{equation}\label{eq:cdf_snr_closed}
{{\rm F}_{\gamma_1^{\rm RS}}(x)}=1 - e^{-\gamma\vartheta_1} \sum_{k=0}^{m_1-1} \frac{(\gamma\vartheta_1)^k}{k!},
\end{equation}
where $\vartheta_{1}(\theta) = m_1 / \bar{\gamma}_1^{\rm RS}(\theta)$.
\end{lemma}

\begin{proof}
See Appendix \ref{plema1}.
\end{proof}

By substituting the CDF from \eqref{eq:cdf_snr_closed} into \eqref{epaprir}, $\overline{\epsilon}_{1}^{\mathrm{RS}}$ can be obtained as
\begin{equation}
\label{eq:bler_integral_start}
\begin{aligned}
\bar{\epsilon}^{\rm RS}_1 &\approx \chi \int_{\rho_L}^{\rho_H} \left(1 - e^{-x\vartheta_1} \sum_{k=0}^{m_1-1} \frac{(x\vartheta_1)^k}{k!}\right) dx \\
&= \chi \left( \int_{\rho_L}^{\rho_H} dx - \sum_{k=0}^{m_1-1} \frac{\vartheta_1^k}{k!} \int_{\rho_L}^{\rho_H} x^k e^{-x\vartheta_1} dx \right).
\end{aligned}
\end{equation}
The integral term in \eqref{eq:bler_integral_start} is evaluated using the definition of the lower incomplete gamma function, $\gamma(s,z)=\int_0^z t^{s-1}e^{-t}dt$, which gives
\begin{equation}
\label{eq:integral_xk_exp}
\int_{\rho_L}^{\rho_H} x^k e^{-x\vartheta_1} dx = \frac{\gamma(k+1, \rho_H\vartheta_1) - \gamma(k+1, \rho_L\vartheta_1)}{\vartheta_1^{k+1}}.
\end{equation}
Substituting \eqref{eq:integral_xk_exp} into \eqref{eq:bler_integral_start}, we obtain
\begin{equation}
\label{eq:bler_intermediate}
\bar{\epsilon}^{\rm RS}_1 \!\approx\! \chi \left( (\rho_H \!-\! \rho_L)\! - \!\sum_{k=0}^{m_1-1} \frac{\gamma(k\!+\!1, \rho_H\vartheta_1)\! - \!\gamma(k\!+\!1, \rho_L\vartheta_1)}{k!\vartheta_1} \right).
\end{equation}
Using the identity $\gamma(k+1, z) = k!(1-e^{-z}\sum_{j=0}^{k} z^j/j!)$ from \cite{abramowitz1972}, \eqref{eq:bler_intermediate} can be obtained as (\ref{eq:bler_final_double_col}) (see top of next page).
\begin{figure*}
\begin{equation}\label{eq:bler_final_double_col}
\bar{\epsilon}^{\rm RS}_1 \approx \chi \left[(\rho_H - \rho_L) - \frac{1}{\vartheta_1}\sum_{k=0}^{m_1-1} \left( e^{-\rho_L\vartheta_1}\sum_{j=0}^{k}\frac{(\rho_L\vartheta_1)^j}{j!}- e^{-\rho_H\vartheta_1}\sum_{j=0}^{k}\frac{(\rho_H\vartheta_1)^j}{j!} \right)\right]
\end{equation}
\hrulefill
\end{figure*}
In this way, we obtain the average BLER $\bar{\epsilon}^{\rm RS}_1$. For $\bar{\epsilon}^{\rm RS}_2$, we first obtain the CDF of $\bar{\epsilon}^{\rm RS}_2$ in Lemma \ref{lema2}.

\subsubsection{Second-Hop Average BLER}
For the FAS-enabled second hop, the average BLER is determined by the statistics of $|h_{\rm FAS}|^2$ defined in \eqref{eq:analytical_model_def}. The CDF of the resulting instantaneous SNR is given in the following lemma.

\begin{lemma}\label{lema2}
For ${{\rm F}_{\gamma_2^{\rm RS}}(x)}$, the CDF can be expressed as
\begin{equation}\label{lem2xy}
{{\rm F}_{\gamma_2^{\rm RS}}(x)} =\prod_{n=1}^{{N_{\rm eff}} } \left(1 - e^{-\frac{x \vartheta_2}{\lambda_n}} \sum_{k=0}^{m_2-1} \frac{1}{k!}\left(\frac{x \vartheta_2}{\lambda_n}\right)^k\right),
\end{equation}
where $\vartheta_{2}(\theta) = m_2 \sum_{n=1}^{N_{\rm eff}}  \lambda_n / \bar{\gamma}_2^{\rm RS}(\theta)$.
\end{lemma}

\begin{proof}
See Appendix \ref{plema2}.
\end{proof}

Substituting \eqref{lem2xy} into \eqref{epaprir}, $\bar{\epsilon}^{\rm RS}_2$ can be evaluated as
\begin{equation}\label{ruaera}
\begin{aligned}
\bar{\epsilon}^{\rm RS}_2
&\approx\chi\int_{\rho_L}^{\rho_H}\prod_{n=1}^{{N_{\rm eff}} } \left(1 - e^{-\frac{x \vartheta_2}{\lambda_n}} \sum_{k=0}^{m_2-1} \frac{1}{k!}\left(\frac{x \vartheta_2}{\lambda_n}\right)^k\right)dx.
\end{aligned}
\end{equation}

To obtain a closed-form solution to \eqref{ruaera}, we apply the exact derivation method. Specifically, for simplicity, we let
\begin{equation}
A_n(x) \triangleq e^{-\frac{x \vartheta_2}{\lambda_n}} \sum_{k=0}^{m_2-1} \frac{1}{k!}\left(\frac{x \vartheta_2}{\lambda_n}\right)^k.
\end{equation}
By applying the principle of inclusion-exclusion, the product can be expanded into a sum of $2^{{N_{\rm eff}} }$ terms, i.e.,
\begin{equation}\label{ruaera1}
\prod_{n=1}^{{N_{\rm eff}} } (1 - A_n(x)) = \sum_{\mathcal{S} \subseteq \{1, \dots, {N_{\rm eff}} \}} (-1)^{|\mathcal{S}|} \prod_{j \in \mathcal{S}} A_j(x),
\end{equation}
where the sum is taken over all possible subsets $\mathcal{S}$ of the index set $\{1, \dots, {N_{\rm eff}} \}$, and $|\mathcal{S}|$ denotes the cardinality of subset $\mathcal{S}$. By convention, if $\mathcal{S} = \emptyset$, the product term is equal to $1$.

Substituting \eqref{ruaera1} into \eqref{ruaera}, and interchanging the order of the finite summation and integration, we obtain
\begin{equation}
\overline{\epsilon}^{\rm RS}_2 = \chi \sum_{\mathcal{S} \subseteq \{1, \dots, {N_{\rm eff}} \}} (-1)^{|\mathcal{S}|} \int_{\rho_L}^{\rho_H}  \prod_{j \in \mathcal{S}} A_j(x)  dx. \label{eq:sum_of_integrals}
\end{equation}
The problem is now reduced to finding a closed-form solution for the integral corresponding to each subset $\mathcal{S}$. For any non-empty subset $\mathcal{S}$, the integrand has the structure
\begin{align}
\prod_{j \in \mathcal{S}} A_j(x) &= \prod_{j \in \mathcal{S}} e^{-\frac{x \vartheta_2}{\lambda_j}} \sum_{k=0}^{m_2-1} \frac{1}{k!}\left(\frac{x \vartheta_2}{\lambda_j}\right)^k\nonumber \\
&= e^{-x \sum\limits_{j \in \mathcal{S}} \frac{\vartheta_2}{\lambda_j}} P_{\mathcal{S}}(x), \label{eq:integrand_structure}
\end{align}
in which $P_{\mathcal{S}}(x) \triangleq \prod_{j \in \mathcal{S}} \sum_{k=0}^{m_2-1} \frac{1}{k!}\left(\frac{x \vartheta_2}{\lambda_j}\right)^k$ is a polynomial in $x$ of degree $|\mathcal{S}|(m-1)$. To construct the final closed-form expression, we first evaluate the term for the empty set $\mathcal{S}=\emptyset$ in (\ref{eq:sum_of_integrals}), which contributes
\begin{equation}
\chi \int_{\rho_L}^{\rho_H} 1 dx = \chi(\rho_H - \rho_L).
\end{equation}
For any non-empty subset $\mathcal{S}$, we define the coefficient
\begin{equation}
b_{\mathcal{S}} \triangleq \sum\limits_{j \in \mathcal{S}} \frac{\vartheta_2}{\lambda_j},
\end{equation}
As a result, we expand the polynomial $P_{\mathcal{S}}(x)$ in its canonical monomial basis as
\begin{equation}
P_{\mathcal{S}}(x) = \sum_{a=0}^{|\mathcal{S}|(m_2-1)} c_a(\mathcal{S}) x^a,
\end{equation}
where $c_a(\mathcal{S})$ are found by expanding the product definition of $P_{\mathcal{S}}(x)$. The integral for each subset $\mathcal{S}$ thus becomes a linear combination of elementary integrals. A standard result from integral calculus states that for a non-negative integer $a$,
\begin{equation}
\int x^a e^{-bx} dx = -\frac{a!}{b^{a+1}} e^{-bx} \sum_{i=0}^a \frac{(bx)^i}{i!}.
\end{equation}
To simplify the final expression, we define a helper function $\mathcal{G}(y; a, b)$ for the evaluation of the definite integral
\begin{equation}
\mathcal{G}(y; a, b) \triangleq -\frac{a!}{b^{a+1}} e^{-by} \sum\nolimits_{i=0}^a \frac{(by)^i}{i!}.
\end{equation}
Combining these results, the exact closed-form expression for the BLER is meticulously constructed as
\begin{multline}\label{eq:bler_final_xscol}
\overline{\epsilon}^{\rm RS}_2 =\chi(\rho_H - \rho_L)+ \\
\chi \sum_{\substack{\mathcal{S} \subseteq \{1, \dots, N_{\rm eff}\} \\ \mathcal{S} \neq \emptyset}} (-1)^{|\mathcal{S}|} \sum_{a=0}^{|\mathcal{S}|(m_2-1)} c_a(\mathcal{S})\\
\times \left[ \mathcal{G}(\rho_H; a, b_{\mathcal{S}}) - \mathcal{G}(\rho_L; a, b_{\mathcal{S}}) \right].
\end{multline}
In this way, we obtain the average BLER $\overline{\epsilon}^{\rm RS}_2$.

To obtain more useful insights, we consider an approximation that is highly accurate in the high SNR regime.

\begin{theorem}\label{theoreapd1}
In high SNR regime, $\bar{\epsilon}^{\rm RS}_2$ can be written as
\begin{multline}
\bar{\epsilon}^{\rm RS}_2 \approx \frac{\chi \left( \rho_H^{m_2 {N_{\rm eff}}  + 1} - \rho_L^{m_2 {N_{\rm eff}}  + 1} \right)}{m_2 {N_{\rm eff}}  + 1}\\
\times \left(\frac{\vartheta_2^{m_2}}{\Gamma(m_2+1)}\right)^{{N_{\rm eff}} }\prod_{n=1}^{{N_{\rm eff}} } \lambda_n^{-m_2}.
\end{multline}
\end{theorem}

\begin{proof}
See Appendix \ref{prooftheoreapd1}.
\end{proof}

\begin{remark}
Theorem \ref{theoreapd1} provides a crucial insight into the system's high-SNR behavior, revealing that the diversity order of the FAS-enabled link is precisely $m_2 N_{\text{eff}}$. This result quantitatively demonstrates that the performance is jointly determined by the small-scale fading severity $m_2$ and the spatial degrees of freedom harvested by the FAS $N_{\text{eff}}$.
\end{remark}

Fixing the UAV position, the average BLER via \eqref{eq:bler_final_double_col} and \eqref{eq:bler_final_xscol} can be given by
\begin{equation}
\overline{\epsilon}_{T}^{\rm RS}(\theta) = 1 - (1-\overline{\epsilon}_{1}^{\rm RS}(\theta))(1-\overline{\epsilon}_{2}^{\rm RS}(\theta)).
\end{equation}
Under our model, as the UAV traverses a circular trajectory of radius $r$ at altitude $Z_3$, the quantity $\bar{\epsilon}_T^{\rm RS}$ depends on the angular parameter $\theta$. Averaging with respect to the heading angle with uniform distribution yields
\begin{equation}
\bar{\epsilon}_O^{\rm RS}
= \int_{0}^{2\pi}\!\bar{\epsilon}_T^{\rm RS}(\theta)\,f_\theta(\theta)\,{\rm d}\theta
= \frac{1}{2\pi}\int_{0}^{2\pi}\!\bar{\epsilon}_T^{\rm RS}(\theta)\,{\rm d}\theta,
\end{equation}
where $f_\theta(\theta)=\frac{1}{2\pi}$ is the density of $\theta$ on $[0,2\pi)$. Since $\bar{\epsilon}_T^{\text{RS}}(\theta)$ does not admit a closed-form integration, we approximate it by $M$-point Gauss-Chebyshev quadrature (GCQ) as
\begin{equation}
\bar{\epsilon}_O^{\rm RS}\;\approx\;\frac{\pi}{M}\sum_{m=1}^{M} w_m\left[\bar{\epsilon}_1^{\rm RS}(\theta_m)
+\left(1-\bar{\epsilon}_1^{\rm RS}(\theta_m)\right)\bar{\epsilon}_2^{\rm RS}(\theta_m)\right],
\end{equation}
where the nodes and weights for Chebyshev quadrature on the interval $[-1, 1]$ are given by
\begin{equation}\label{overbm}
\left\{\begin{aligned}
 \theta_m &= \pi \,x_m + \pi,\\
 w_m &= \frac{\pi}{M} \, \frac{1}{\sqrt{1-x_m^2}},~ m = 1, \ldots, M,
\end{aligned}\right.
\end{equation}
with the nodes $x_m$ as the Chebyshev roots on $[-1, 1]$.

\subsection{Urban Scenario}
Under the urban scenario, the instantaneous BLER $\bar{\epsilon}^{k}_i$ can be given by
\begin{equation}
\epsilon_i^k\approx Q\left(\frac{C(\gamma^k_i)-R}{\sqrt{Z(\gamma^k_i)/L}}\right).
\end{equation}
Considering a certain UAV location, the average BLER for the $k$-link in phase $i$ can be written as
\begin{equation}
\bar{\epsilon}_i^k\approx \int_0^\infty Q\left(\frac{C(\gamma^k_i)-R}{\sqrt{Z(\gamma^k_i)/L}}\right) f_{\gamma_i^k}(x)dx,
\end{equation}
where $f_{\gamma_i^k}(x)$ is the PDF of $\gamma_i^k$.

To address this problem, we follow the approach outlined in \eqref{epaprir} in the rural scenario. Next, we use Lemmas \ref{lems3} and \ref{lems4} to find the CDFs of the first and second phases.

\begin{lemma}\label{lems3}
For the first hop, the CDF of the instantaneous SNR $\gamma_1^k$ for a link type $k \in \{\mathrm{LoS}, \mathrm{NLoS}\}$ is given by
\begin{equation}
F_{\gamma_{1}^{k}}(x)=1-e^{-x\vartheta_{1}^{k}}\sum\nolimits_{j=0}^{m_1-1}\frac{(x\vartheta_{1}^{k})^{j}}{j!},
\end{equation}
where we have the parameter $\vartheta_1^k(\theta) \triangleq m_1 / \bar{\gamma}_1^k(\theta)$.
\end{lemma}

\begin{proof}
The proof follows the same steps as in Appendix \ref{plema1}, with the rural path loss $\beta_1$ replaced by $\beta_{1}^{k}(\theta)$.
\end{proof}

\begin{lemma}\label{lems4}
For the FAS-enabled second hop, the CDF of the instantaneous SNR $\gamma_2^k$ for a link type $k \in \{\mathrm{LoS}, \mathrm{NLoS}\}$ is
\begin{equation}\label{ga2kxl4}
F_{\gamma_{2}^{k}}(x)=\prod_{n=1}^{{N_{\rm eff}} }\left(1-e^{-\frac{x\vartheta_{2}^k}{\lambda_{n}}}\sum_{j=0}^{m_2^k-1}\frac{1}{j!}\left(\frac{x\vartheta_{2}^k}{\lambda_{n}}\right)^{j}\right),
\end{equation}
where $m_2^k$ is the Nakagami-m parameter for the corresponding link type and $\vartheta_{2}^{k}(\theta) \triangleq m_2^k \sum_{n=1}^N \lambda_n / \bar{\gamma}_2^k(\theta)$.
\end{lemma}

\begin{proof}
The proof follows the same steps as in Appendix \ref{plema2}, using the appropriate Nakagami-$m$ parameter $m_2^k$ and urban path loss $\beta_{2}^{k}(\theta)$ for each case.
\end{proof}

By substituting the CDF from Lemma~\ref{lems3} into the integral form of the average BLER, we arrive at the closed-form expressions for the first and second hops as follows.

\subsubsection{First-Hop Average BLER}
For the first hop, the average BLER for a link type $k \in \{\mathrm{LoS}, \mathrm{NLoS}\}$, denoted $\bar{\epsilon}_{1}^{k}(\theta)$, is found by solving the integral. Similar to~\eqref{eq:bler_final_double_col}, we have (\ref{eq:urban_hop1_k}) (see top of next page),
\begin{figure*}
\begin{equation}\label{eq:urban_hop1_k}
\overline{\epsilon}_{1}^{k}(\theta) \approx \chi \left[(\rho_H - \rho_L) - \frac{1}{\vartheta_1^k} \sum_{j=0}^{m_1-1} \left( e^{-\rho_L\vartheta_1^k}\sum_{l=0}^{j}\frac{(\rho_L\vartheta_1^k)^l}{l!}- e^{-\rho_H\vartheta_1^k}\sum_{l=0}^{j}\frac{(\rho_H\vartheta_1^k)^l}{l!} \right)\right]
\end{equation}
\hrulefill
\end{figure*}
where the parameter $\vartheta_{1}^{k}(\theta)$ is defined in Lemma~\ref{lems3}. The total average BLER for the first hop is then the expectation over the link probabilities
\begin{equation}
    \bar{\epsilon}_{1}^{\mathrm{US}}(\theta) = \bar{\epsilon}_{1}^{\mathrm{LoS}}(\theta)P_{1}^{\mathrm{LoS}}(\theta) + \bar{\epsilon}_{1}^{\mathrm{NLoS}}(\theta)P_{1}^{\mathrm{NLoS}}(\theta). \label{eq:urban_hop1_final}
\end{equation}

\subsubsection{Second-Hop Average BLER}
For the FAS-enabled second hop, the derivation for each link type follows the intricate process outlined for the rural case, which involves the principle of inclusion-exclusion. This results in
\begin{multline}\label{eq:urban_hop2_k}
\overline{\epsilon}_{2}^{k}(\theta)\approx{} \chi(\rho_H - \rho_L) +\\
\chi \sum_{\substack{\mathcal{S} \subseteq \{1, \dots, N_{\rm eff}\} \\
\mathcal{S} \neq \emptyset}} (-1)^{|\mathcal{S}|} \times\\
\sum_{a=0}^{|\mathcal{S}|(m_2^k-1)} c_a^k(\mathcal{S}) \left[ \mathcal{G}(\rho_H; a, b_{\mathcal{S}}^k) - \mathcal{G}(\rho_L; a, b_{\mathcal{S}}^k) \right],
\end{multline}
where $c_a^k(\mathcal{S})$ and $b_{\mathcal{S}}^k$ are calculated based on the link-specific parameter $\vartheta_{2}^{k}(\theta)$ and $m_2^k$ as defined in Lemma~\ref{lems4}.

In the high SNR regime, we can obtain a highly accurate approximation to yield more useful insights.

\begin{theorem}\label{theo2prd}
In the high SNR regime, the average BLER of the second hop for a given link type $k \in \{\mathrm{LoS}, \mathrm{NLoS}\}$, denoted $\overline{\epsilon}_{2}^{k}$, can be approximated by
\begin{multline}
\bar{\epsilon}_2^k \approx \frac{\chi \left( \rho_H^{m_2 {N_{\rm eff}}  + 1} - \rho_L^{m_2 {N_{\rm eff}}  + 1} \right)}{m_2 {N_{\rm eff}}  + 1} \\
\times\left(\frac{\vartheta_2^{m_2}}{\Gamma(m+1)}\right)^{{N_{\rm eff}} } \prod_{n=1}^{{N_{\rm eff}} } \lambda_n^{-m_2}.
\end{multline}
\end{theorem}

\begin{proof}
The proof follows the same methodology as in Appendix \ref{prooftheoreapd1}, by applying the high-SNR approximation to the CDF in Lemma \ref{lems4} for a specific link type $k$.
\end{proof}

\begin{remark}
Theorem \ref{theo2prd} reveals that when a specific link condition (LoS or NLoS) is maintained, the FAS link behaves as a power-limited system, where increasing transmit power continuously reduces the BLER at a rate determined by the diversity order $m_2^k N_{\rm eff}$. This highlights the intrinsic capability of FAS to combat fading under a fixed channel type.
\end{remark}

Given the angle of UAV, the BLER for the second hop is subsequently given by
\begin{equation}\label{eq:urban_hop2_final}
\bar{\epsilon}_{2}^{\mathrm{US}}(\theta) = \bar{\epsilon}_{2}^{\mathrm{LoS}}(\theta)P_{2}^{\mathrm{LoS}}(\theta) + \bar{\epsilon}_{2}^{\mathrm{NLoS}}(\theta)P_{2}^{\mathrm{NLoS}}(\theta).
\end{equation}

\subsubsection{End-to-End BLER and Asymptotic Analysis}
The average BLER for each hop is now established for any given UAV location $\theta$. We proceed to combine these results to determine the overall system performance. The end-to-end average BLER for the urban scenario, $\bar{\epsilon}_{T}^{\mathrm{US}}(\theta)$, is obtained by substituting the hop-level results from~\eqref{eq:urban_hop1_final} and~\eqref{eq:urban_hop2_final} into the standard DF BLER formula, yielding
\begin{equation}
\begin{aligned}
\overline{\epsilon}_{T}^{\mathrm{US}}(\theta)
% =& 1 - (1-\overline{\epsilon}_{1}^{\mathrm{US}}(\theta))(1-\overline{\epsilon}_{2}^{\mathrm{US}}(\theta))\\
% =&\overline{\epsilon}_{1}^{\mathrm{US}}(\theta)+ \overline{\epsilon}_{2}^{\mathrm{US}}(\theta)           -\overline{\epsilon}_{1}^{\mathrm{US}}(\theta)\overline{\epsilon}_{2}^{\mathrm{US}}(\theta)\\
&=\bar{\epsilon}_{1}^{\mathrm{LoS}}(\theta)P_{1}^{\mathrm{LoS}}(\theta) + \bar{\epsilon}_{1}^{\mathrm{NLoS}}(\theta)P_{1}^{\mathrm{NLoS}}(\theta)\\
&+\bar{\epsilon}_{2}^{\mathrm{LoS}}(\theta)P_{2}^{\mathrm{LoS}}(\theta)+ \bar{\epsilon}_{2}^{\mathrm{NLoS}}(\theta)P_{2}^{\mathrm{NLoS}}(\theta)\\
&-[\bar{\epsilon}_{1}^{\mathrm{LoS}}(\theta)P_{1}^{\mathrm{LoS}}(\theta) + \bar{\epsilon}_{1}^{\mathrm{NLoS}}(\theta)P_{1}^{\mathrm{NLoS}}(\theta)]\\
&\times[\bar{\epsilon}_{2}^{\mathrm{LoS}}(\theta)P_{2}^{\mathrm{LoS}}(\theta)+ \bar{\epsilon}_{2}^{\mathrm{NLoS}}(\theta)P_{2}^{\mathrm{NLoS}}(\theta)].
\end{aligned}
\end{equation}
The overall average BLER, $\bar{\epsilon}_{O}^{\mathrm{US}}$, is then found by averaging $\bar{\epsilon}_{T}^{\mathrm{US}}(\theta)$ over the UAV's circular trajectory with respect to the uniform distribution of  $\theta$ as
% \begin{equation}
% \bar{\epsilon}_{O}^{\mathrm{US}} = \frac{1}{2\pi}\int_{0}^{2\pi}\bar{\epsilon}_{T}^{\mathrm{US}}(\theta)d\theta.
% \end{equation}
\begin{equation}
\begin{aligned}
\overline{\epsilon}_{O}^{\mathrm{US}}
&= \frac{1}{2\pi} \int_{0}^{2\pi} \Big\{ \left[\overline{\epsilon}_{1}^{\mathrm{LoS}}(\theta)P_{1}^{\mathrm{LoS}}(\theta)+\overline{\epsilon}_{1}^{\mathrm{NLoS}}(\theta)P_{1}^{\mathrm{NLoS}}(\theta)\right] \\
& + \left[\overline{\epsilon}_{2}^{\mathrm{LoS}}(\theta)P_{2}^{\mathrm{LoS}}(\theta)+\overline{\epsilon}_{2}^{\mathrm{NLoS}}(\theta)P_{2}^{\mathrm{NLoS}}(\theta)\right] \\
&- \left[\overline{\epsilon}_{1}^{\mathrm{LoS}}(\theta)P_{1}^{\mathrm{LoS}}(\theta)+\overline{\epsilon}_{1}^{\mathrm{NLoS}}(\theta)P_{1}^{\mathrm{NLoS}}(\theta)\right]\\
&\times \left[\overline{\epsilon}_{2}^{\mathrm{LoS}}(\theta)P_{2}^{\mathrm{LoS}}(\theta)+\overline{\epsilon}_{2}^{\mathrm{NLoS}}(\theta)P_{2}^{\mathrm{NLoS}}(\theta)\right] \Big\} d\theta.
\end{aligned}
\end{equation}
As this integral does not admit a closed-form solution, again, we employ the $M$-point GCQ method for a precise numerical approximation as
\begin{equation}
\bar{\epsilon}_{O}^{\mathrm{US}}\;\approx\;\frac{\pi}{M}\sum_{m=1}^{M} w_m\, \bar{\epsilon}_{T}^{\mathrm{US}}(\theta_m),
\end{equation}
where the nodes $\theta_m$ and weights $w_m$ are defined in~\eqref{overbm}.

\begin{theorem}\label{theo:urban_system_asymp}
The end-to-end performance in the urban scenario is fundamentally limited by an error floor when the UAV transmit power is sufficiently large. Let $\mathbb{E}_{\theta}[\cdot]$ denote the expectation over the uniform distribution of the UAV's heading angle $\theta \in [0, 2\pi)$. For a fixed BS power $P_1$ and $P_2 \to \infty$, the overall average BLER converges to a floor determined solely by the performance of the first hop
\begin{align}
&\lim_{P_2 \to \infty} \overline{\epsilon}_{O}^{\mathrm{US}} = \frac{\int_{0}^{2\pi} \left[ \overline{\epsilon}_{1}^{\mathrm{LoS}}(\theta)P_{1}^{\mathrm{LoS}}(\theta) + \overline{\epsilon}_{1}^{\mathrm{NLoS}}(\theta)P_{1}^{\mathrm{NLoS}}(\theta) \right] d\theta}{2\pi} . \label{eq:corrected_floor}
\end{align}
\end{theorem}

\begin{proof}
The proof follows from the asymptotic behavior of the two-hop DF relaying system. For the first hop, with fixed power $P_1$, the average BLER $\overline{\epsilon}_{1}^{\mathrm{US}}(\theta; P_1)$ remains constant regardless of $P_2$. For the second hop, as $P_2 \to \infty$, the transmit power becomes sufficient to overcome the path loss of both LoS and NLoS links. Thus, the instantaneous SNR $\gamma_2^k \to \infty$ for any channel state $k \in \{{\rm LoS},{\rm NLoS}\}$, which implies that the BLER of the second hop converges to zero, i.e., $\overline{\epsilon}_{2}^{\mathrm{US}}(\theta) \to 0$.

The end-to-end BLER for a given $\theta$ is given by
\begin{equation}
\overline{\epsilon}_{T}^{\mathrm{US}}(\theta) = 1 - (1 - \overline{\epsilon}_{1}^{\mathrm{US}}(\theta))(1 - \overline{\epsilon}_{2}^{\mathrm{US}}(\theta)).
\end{equation}
Taking the limit as $P_2 \to \infty$, we have
\begin{equation}
\lim_{P_2 \to \infty} \overline{\epsilon}_{T}^{\mathrm{US}}(\theta) = 1 - (1 - \overline{\epsilon}_{1}^{\mathrm{US}}(\theta))(1 - 0) = \overline{\epsilon}_{1}^{\mathrm{US}}(\theta).
\end{equation}
Averaging over the distribution of $\theta$ completes the proof.
\end{proof}

\begin{remark}
Theorem \ref{theo:urban_system_asymp} reveals a fundamental performance bottleneck in urban UAV relaying systems. The floor provides a crucial insight: for a system with a fixed backhaul (BS-to-UAV link), the overall performance is fundamentally limited by the reliability of that backhaul link. This implies that increasing the UAV's power ($P_2$) yields diminishing returns; once $P_2$ is sufficiently high to make the second hop reliable, any further increase in power is completely ineffective at improving the end-to-end BLER. Therefore, to improve performance beyond this floor, system optimization must focus on enhancing the first hop, for instance by increasing $P_1$ or improving the BS-to-UAV channel through trajectory planning.
\end{remark}

\section{EE Maximization}\label{sec:eemax}
In this section, we develop a framework to maximize the EE of the FAS-enabled UAV system under investigation. We formulate a realistic EE model that captures the fundamental trade-off between the diversity gain afforded by FAS and the operational cost incurred during port selection \cite{wong2021fluid, zhang20205ee}. While a larger number of ports $N$ enhances diversity, it also introduces time and energy overheads for channel estimation and switching, suggesting that an optimal $N$ exists.

\subsection{Problem Formulation}
To capture the trade-off, we define EE as the number of successfully delivered bits per unit of energy (bits/Joule). The total energy consumed by the UAV during one transmission block, $E_{\text{total}}$, is composed of three components: transmit energy, static circuit energy, and FAS switching energy.

The total block duration is denoted by $T_{\text{block}} = L/W_{\text{band}}$, where $L$ is the blocklength and $W_{\text{band}}$ is the system bandwidth. The port selection process for $N$ ports incurs a time overhead of $T_{\text{sw}}(N) = N\tau_p$, where $\tau_p$ is the processing time per port. Thus, the effective data transmission time is $T_{\text{tx}}(N) = T_{\text{block}} - T_{\text{sw}}(N)$. Let $P_2$, $P_c$, and $P_{\text{sw}}$ denote the UAV's transmit power, static circuit power, and the additional power consumed during port selection, respectively. We have the total energy
\begin{equation}
E_{\text{total}} = P_2 T_{\text{tx}}(N) + P_c T_{\text{block}} + P_{\text{sw}} T_{\text{sw}}(N).
\end{equation}
The number of successfully delivered bits is $B_{\text{succ}} = B(1 - \bar{\epsilon}_O)$. Therefore, the EE is formulated as
\begin{equation}\label{eq:ee_model}
\text{EE} = \frac{B(1 - \bar{\epsilon}_O)}{P_2(L/W_{\text{band}} \!-\! N\tau_p)\! + \!P_c(L/W_{\text{band}}) + P_{\text{sw}}N\tau_p},
\end{equation}
where $\bar{\epsilon}_O$ applies to both $\bar{\epsilon}_O^{\rm RS}$ and $\bar{\epsilon}_O^{\rm US}$.

Our objective is to maximize the EE by jointly optimizing the blocklength $L$, UAV altitude $Z_{\rm U}$, transmit power $P_2$, and the number of FAS ports $N$. The optimization problem can be formulated as
\begin{subequations}\label{eq:p1}
\begin{align}
({\rm P}1): \max_{L, Z_U, P_2, N} \quad & \text{EE}(L, Z_U, P_2, N) \\
\text{s.t.} \quad & \overline{\epsilon}_O(L, Z_U, P_2, N) \le \epsilon_{\text{th}}, \label{eq:p1_c1} \\
& 0 < P_2 \le P_{\text{max}}, \label{eq:p1_c2} \\
& Z_{\text{min}} \le Z_U \le Z_{\text{max}}, \label{eq:p1_c3} \\
& L_{\text{min}} \le L \le L_{\text{max}}, \label{eq:p1_c4} \\
& N_{\text{min}} \le N \le N_{\text{max}}, \label{eq:p1_c5} \\
& N\tau_p < L/W_{\text{band}}. \label{eq:p1_c6}
\end{align}
\end{subequations}
The constraint in \eqref{eq:p1_c1} ensures communication reliability, while \eqref{eq:p1_c2}--\eqref{eq:p1_c5} define the operational ranges of the variables, and \eqref{eq:p1_c6} is a crucial causality constraint.

\begin{algorithm}[t]
\caption{Bisection Algorithm for Minimum Power}\label{alg:bisection}
\begin{algorithmic}[1]
\State \textbf{Input:} Fixed parameters $\hat{L}, \hat{Z}_U, \hat{N}, P_{\max}$, target BLER $\epsilon_{\text{th}}$, tolerance $\delta$.
\State \textbf{Output:} Minimum required power $P_2^*$.
\State Initialize: $P_{\text{low}} \leftarrow P_{\min}$, $P_{\text{high}} \leftarrow P_{\max}$.
\If{$\overline{\epsilon}_{O}(P_{\max}; \hat{L}, \hat{Z}_U, \hat{N}) > \epsilon_{\text{th}}$}
    \State \Return Infeasible
\EndIf
\While{$(P_{\text{high}} - P_{\text{low}}) > \delta$}
    \State $P_{\text{mid}} \leftarrow (P_{\text{low}} + P_{\text{high}})/2$.
    \If{$\overline{\epsilon}_{O}(P_{\text{mid}}; \hat{L}, \hat{Z}_U, \hat{N}) > \epsilon_{\text{th}}$}
        \State $P_{\text{low}} \leftarrow P_{\text{mid}}$.
    \Else
        \State $P_{\text{high}} \leftarrow P_{\text{mid}}$.
    \EndIf
\EndWhile
\State \Return $P_{\text{high}}$.
\end{algorithmic}
\end{algorithm}

\begin{algorithm}[t]
\caption{Joint Optimization for Maximum EE}
\label{alg:joint_opt}
\begin{algorithmic}[1]
\State \textbf{Input:} Search ranges $[L_{\min}, L_{\max}]$, $[Z_{\min}, Z_{\max}]$, $[N_{\min}, N_{\max}]$; $P_{\max}, \epsilon_{\text{th}}$; $P_c, P_{sw}, \tau_p, W_{\text{band}}$.
\State \textbf{Output:} Optimal parameters $(L^*, Z_U^*, N^*, P_2^*)$ and $\text{EE}_{\max}$.
\State Initialize: $\text{EE}_{\max} \leftarrow 0$; $(L^*, Z_U^*, N^*, P_2^*) \leftarrow \text{null}$.
\For{each $L \in [L_{\min}, L_{\max}]$}
    \For{each $Z_U \in [Z_{\min}, Z_{\max}]$}
        \State $\text{EE}_{\text{cand}} \leftarrow 0$; $N_{\text{cand}} \leftarrow \text{null}$; $P_{2,\text{cand}} \leftarrow \text{null}$.
        \For{each $N \in [N_{\min}, N_{\max}]$}
            \If{$N\tau_p \ge L/W_{\text{band}}$} \textbf{continue}; \EndIf
            \State $p_{\text{curr}} \leftarrow \text{Algorithm \ref{alg:bisection}}(L, Z_U, N, P_{\max}, \epsilon_{\text{th}})$.
            \If{$p_{\text{curr}}$ is feasible}
                \State Calculate $\text{EE}_{\text{curr}}$ via \eqref{eq:ee_model} with $P_2 \leftarrow p_{\text{curr}}$.
                \If{$\text{EE}_{\text{curr}} > \text{EE}_{\text{cand}}$}
                    \State $\text{EE}_{\text{cand}} \leftarrow \text{EE}_{\text{curr}}$.
                    \State $N_{\text{cand}} \leftarrow N$; $P_{2,\text{cand}} \leftarrow p_{\text{curr}}$.
                \EndIf
            \EndIf
        \EndFor
        \If{$\text{EE}_{\text{cand}} > \text{EE}_{\max}$}
            \State $\text{EE}_{\max} \leftarrow \text{EE}_{\text{cand}}$.
            \State $(L^*, Z_U^*, N^*, P_2^*) \leftarrow (L, Z_U, N_{\text{cand}}, P_{2,\text{cand}})$.
        \EndIf
    \EndFor
\EndFor
\State \Return $(L^*, Z_U^*, N^*, P_2^*), \text{EE}_{\max}$.
\end{algorithmic}
\end{algorithm}

\subsection{Methodology}
Unfortunately, Problem (P1) is a non-convex mixed-integer nonlinear program due to the integer variable $N$ and the non-convex objective function and constraints with respect to the joint variables. Such problems are NP-hard and cannot be solved directly for the global optimum. We therefore propose a hierarchical decomposition approach that breaks (P1) into a sequence of more manageable subproblems.

\subsubsection{Transmit Power Minimization}
For any fixed set of parameters $(\hat{L}, \hat{Z}_U, \hat{N})$, the first subproblem is to find the minimum UAV transmit power $P_2^*$ that satisfies the reliability constraint. This is formulated as
\begin{subequations}
\begin{align}
({\rm P}1.1): \min_{P_2} & \quad P_2 \\
\mathrm{s.t.} & \quad \overline{\epsilon}_{O}(P_2; \hat{L}, \hat{Z}_U, \hat{N}) \le \epsilon_{\rm th}, \\
&\quad  0 < P_2 \le P_{\max}.
\end{align}
\end{subequations}
Leveraging the monotonic relationship between $P_2$ and the BLER, this subproblem can be efficiently solved via a bisection search, as detailed in Algorithm \ref{alg:bisection}.

\subsubsection{Optimal Port Number Determination}
For a fixed blocklength $\hat{L}$ and altitude $\hat{Z}_U$, we find the optimal number of ports $N^*$ that maximizes EE by
\begin{subequations}
\begin{align}
({\rm P}1.2): \max_{N} \quad & \text{EE}(N, P_2^*(N); \hat{L}, \hat{Z}_U) \\
\mathrm{s.t.} \quad & N_{\min} \le N \le N_{\max}, \\
& N\tau_p < \hat{L}/W_{\text{band}},
\end{align}
\end{subequations}
where $P_2^*(N)$ is obtained from solving (P1.1). Since $N$ is an integer, this is solved via a simple search method.

\subsubsection{Optimal Flying Height Determination}
For a fixed blocklength $\hat{L}$, the optimal altitude $Z_U^*$ is found by
\begin{subequations}
\begin{align}
({\rm P}1.3): \max_{Z_U} & \quad \text{EE}^*(Z_U; \hat{L}) \\
\mathrm{s.t.} & \quad  Z_{\min} \le Z_U \le Z_{\max},
\end{align}
\end{subequations}
where $\text{EE}^*(Z_U; \hat{L})$ is the maximum EE obtained from solving (P1.2). This is also solved via a simple search method.

\subsubsection{Overall Solution}
The final solution to (P1) is found by performing an exhaustive search over the blocklength $L$. For each candidate $L$, we solve (P1.3) to find the corresponding maximum EE. The overall optimal solution is the set of parameters that yields the highest EE across all considered blocklengths. This procedure is outlined in Algorithm \ref{alg:joint_opt}.

\begin{table}[t]
\caption{Unified Simulation Parameters for Rural and Urban Scenarios}\label{tab:sim_params_unified}
\centering
\scriptsize
\setlength{\tabcolsep}{3pt}
\renewcommand{\arraystretch}{1.12}

\resizebox{.9\columnwidth}{!}{
\begin{tabular}{|l|c|c|c|}
\hline
\textbf{Parameter} & \textbf{Symbol} & \textbf{Rural Scenario} & \textbf{Urban Scenario} \\ \hline\hline
Data bits per packet          & $F$                  & \multicolumn{2}{c|}{$80\ \text{bits}$} \\ \hline
Carrier frequency             & $f_c$                & \multicolumn{2}{c|}{$2.5\ \text{GHz}$} \\ \hline
Noise power                   & $\sigma^{2}$         & \multicolumn{2}{c|}{$-100\ \text{dBm}$} \\ \hline
System bandwidth              & $W_{\rm band}$                  & \multicolumn{2}{c|}{$10\ \text{MHz}$} \\ \hline
UAV flying radius             & $r$                  & \multicolumn{2}{c|}{$50\ \text{m}$} \\ \hline
Static circuit power             &$P_c$ &\multicolumn{2}{c|} {5 dBm} \\ \hline
FAS switching power            & $P_{\rm sw}$ &\multicolumn{2}{c|} {0 dBm} \\ \hline
Port processing time             & $\tau_p$                   &\multicolumn{2}{c|} {2 $\mu$s} \\ \hline
The number of ports             & $N$                   &\multicolumn{2}{c|} {2 } \\ \hline
Aperture of length             & $W$                   &\multicolumn{2}{c|}{0.5$\lambda$ }\\ \hline
BLER threshold            & $\epsilon_{\rm th}$                   &\multicolumn{2}{c|}{$\leq 10^{-3}$ }\\ \hline
BS coordinates                & $(X_B, Y_B, Z_B)$    & $(1000, 0, 40)\ \text{m}$  & $(100, 0, 40)\ \text{m}$ \\ \hline
UE coordinates                & $(X_E, Y_E, Z_E)$    & $(-1000, 1000, 0)\ \text{m}$ & $(-100, 100, 0)\ \text{m}$ \\ \hline
Nakagami--$m$ (Hop 1) & $m_1$               & $7$  & --- \\ \hline
Nakagami--$m$ (Hop 2) & $m_2$               & $7$  & --- \\ \hline
%LoS probability constants     & $(a,b)$              & ---  & $12.08,\;0.11$ \\ \hline
Additional LoS path loss      & $\eta_{\text{LoS}}$  & ---  & $1.6\ \text{dB}$ \\ \hline
Additional NLoS path loss     & $\eta_{\text{NLoS}}$ & ---  & $23\ \text{dB}$ \\ \hline
NLoS Nakagami--$m$ factor     & $m_{\text{NLoS}}$    & ---  & $1$ (Rayleigh) \\ \hline
\end{tabular}
}
\end{table}

\section{Numerical Results and Discussion}\label{sec:results}
In this section, we present numerical results to validate our analytical framework, quantify the performance gains of the FAS, and the proposed EE maximization. Unless specified otherwise, the simulation parameters are listed in Table \ref{tab:sim_params_unified}.

\begin{figure}[t]
\centering
\subfloat[Rural scenario]{%
\includegraphics[width=0.24\textwidth]{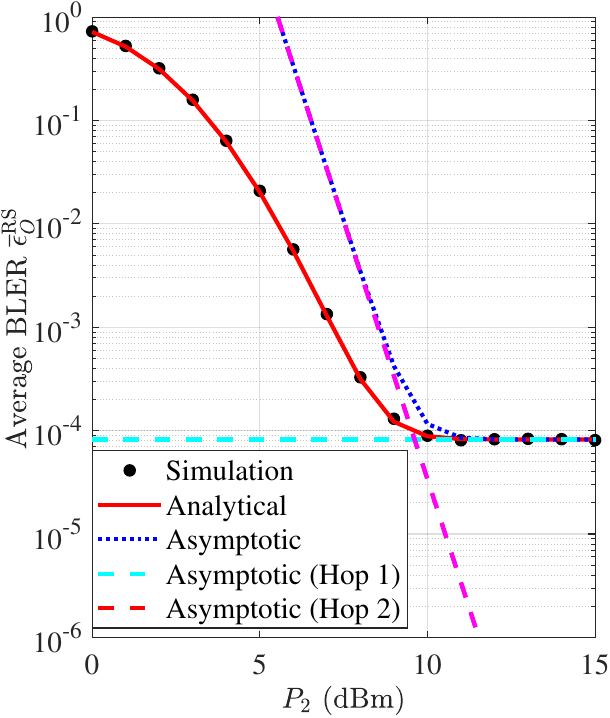}}
\hfil
\subfloat[Urban scenario]{%
\includegraphics[width=0.24\textwidth]{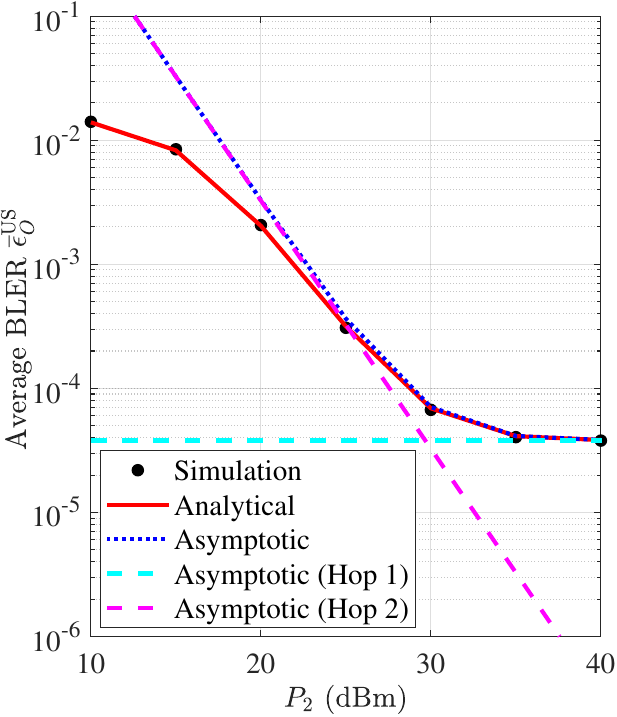}}
\caption{Validation of the analytical framework against Monte Carlo simulation of the analytical model. Parameters: $L=100, N=2, W=0.5\lambda$. (a) Rural scenario: $m_1=m_2=5, P_1=10$ dBm. (b) Urban scenario: $m_{\text{LoS}}=5, m_{\text{NLoS}}=1, P_1=40$ dBm.}\label{fig:validation}
\vspace{-3mm}
\end{figure}

%\subsection{Average BLER Analysis}
We begin by validating our theoretical framework in Fig.~\ref{fig:validation}. The results in this figure compare our derived closed-form expressions for the end-to-end BLER  against Monte Carlo simulations of the underlying analytical model defined in Section \ref{ssec:model}. The perfect alignment between the curves and simulation points confirms the correctness and accuracy of our complex mathematical derivations, for both the rural and urban scenarios. Furthermore, the high-SNR asymptotic expressions are shown to capture the diversity slope in the high-SNR regime accurately. A notable phenomenon observable is the error floor occurred in both scenarios. This happens because as the UAV's transmit power $P_2$ increases, the second-hop BLER diminishes, causing the end-to-end BLER to converge to the BLER of the first hop, which acts as a performance bottleneck. This confirms that our framework correctly models the behavior of the UAV DF-relaying system.

\begin{figure}[t]
\centering
\subfloat[Rural scenario]{%
\includegraphics[width=0.24\textwidth]{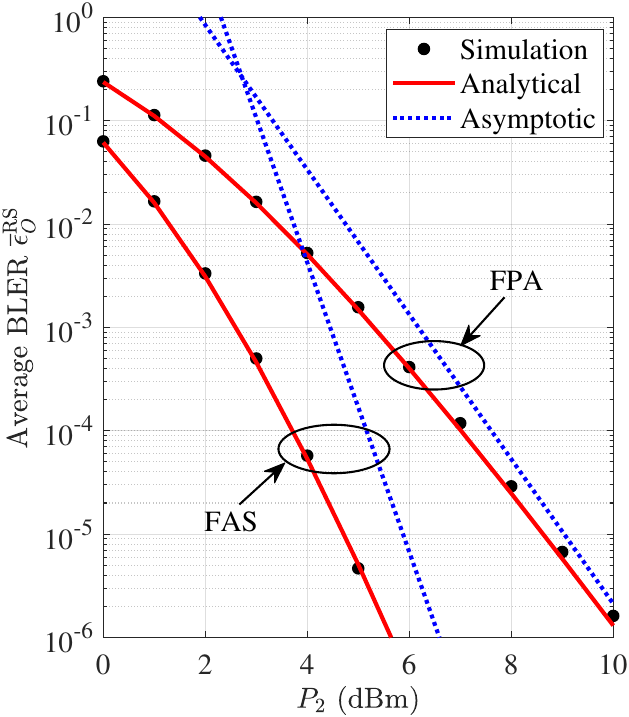}}
\hfil
\subfloat[Urban scenario]{%
\includegraphics[width=0.24\textwidth]{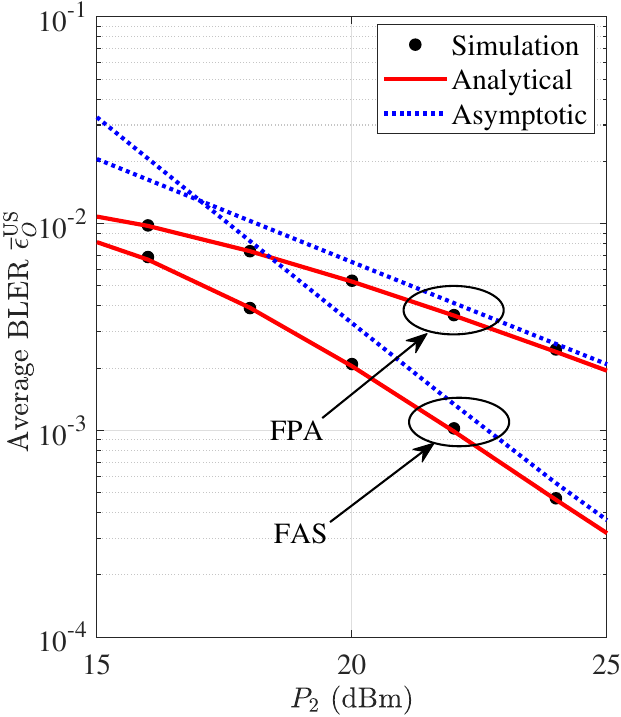}}
\caption{Performance comparison between FAS ($N=2$) and FPA ($N=1$). Parameters are identical to those in Fig.~\ref{fig:validation}.}\label{fig:fas_vs_fpa}
\vspace{-3mm}
\end{figure}

\begin{figure}[t]
\centering
\subfloat[Rural scenario ]{%
\includegraphics[width=0.24\textwidth]{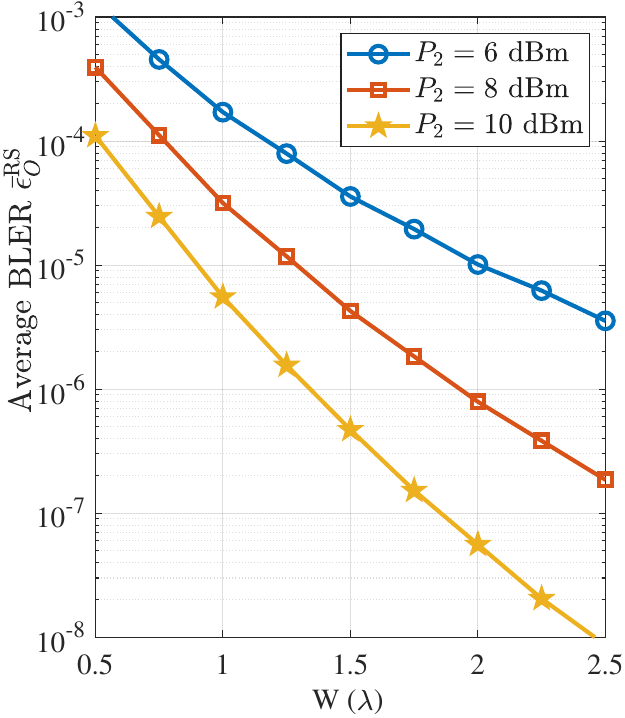}}
\hfil
\subfloat[Urban scenario ]{%
\includegraphics[width=0.24\textwidth]{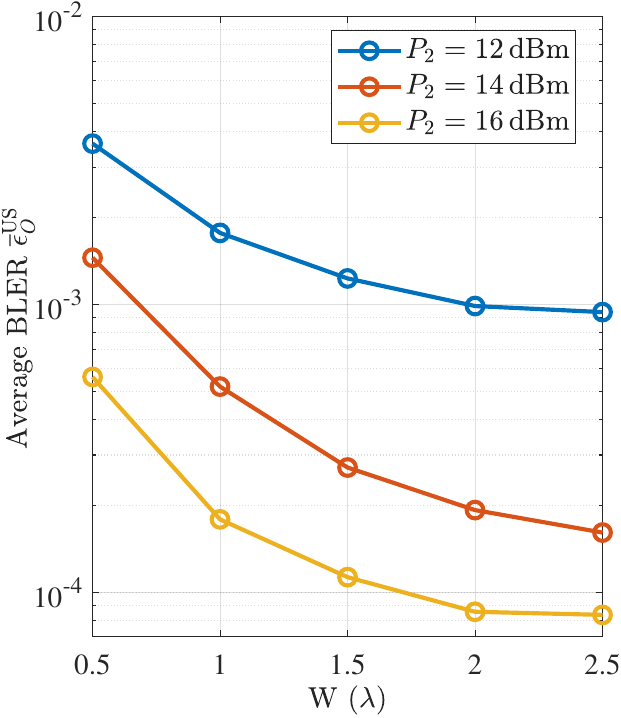}}
\caption{End-to-end BLER versus FAS aperture size $W$ for a fixed $N=8$. (a) Rural scenario: $L=200, m_1=5, m_2=7, P_1=15$ dBm. (b) Urban scenario: $L=100, m_{\text{LoS}}=5, m_{\text{NLoS}}=1, P_1=40$ dBm.}\label{fig:bler_vs_w}
\vspace{-3mm}
\end{figure}

Fig.~\ref{fig:fas_vs_fpa} compares the BLER of a 2-port FAS with a conventional FPA, which is equivalent to $N=1$. First, the figure clearly validates the accuracy of our theoretical derivations, as the Monte Carlo simulation results are in close agreement with our analytical results. Additionally, in the high SNR regime, the asymptotic analysis accurately reflects the system's performance trend. Second, the figure compares the performance of our proposed FAS scheme with the traditional FPA baseline. The steeper slope of the FAS curve confirms its higher diversity order, enabling it to better combat fading. The results demonstrate that due to the superior spatial diversity gain afforded by FAS, the proposed scheme significantly outperforms the FPA scheme in both scenarios. Finally, the figure visually illustrates the performance difference between rural and urban scenarios. In the rural (LoS-dominant) scenario, the system performs considerably better due to more favorable channel conditions compared to the urban (NLoS) scenario.

In Fig.~\ref{fig:bler_vs_w}, we study the impact of the FAS aperture size $W$ on BLER for a fixed number of ports ($N=8$). In both scenarios, the BLER decreases as $W$ increases because a larger aperture reduces the spatial correlation between the ports, thereby enhancing the effective diversity gain harvested by the FAS. However, this improvement clearly exhibits diminishing returns. As $W$ becomes sufficiently large (e.g., $W > 2\lambda$), the channels at the ports become sufficiently decorrelated, and further increasing the aperture yields only marginal performance gains. This saturation effect reveals a crucial design insight: the FAS aperture must be large enough to ensure effective diversity, but an excessively large physical footprint is unnecessary.

\begin{figure}[t]
\centering
\subfloat[Rural scenario ]{%
\includegraphics[width=0.24\textwidth]{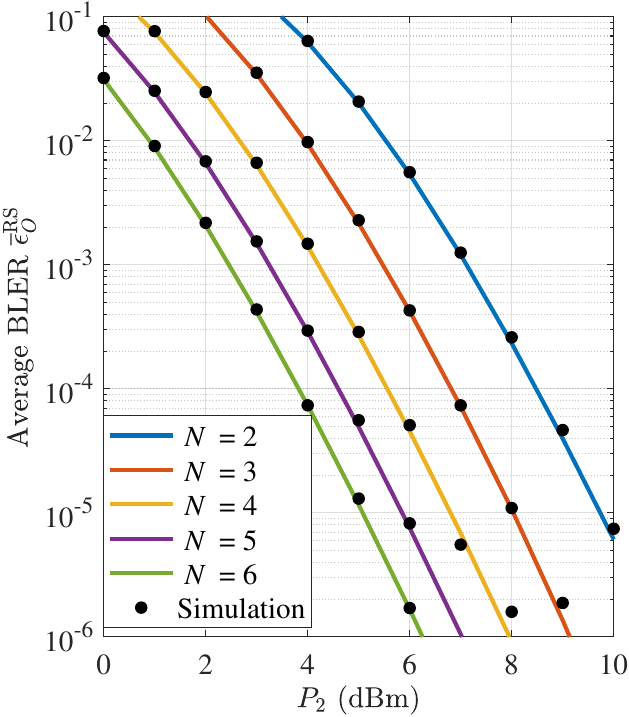}}
\hfil
\subfloat[Urban scenario ]{%
\includegraphics[width=0.24\textwidth]{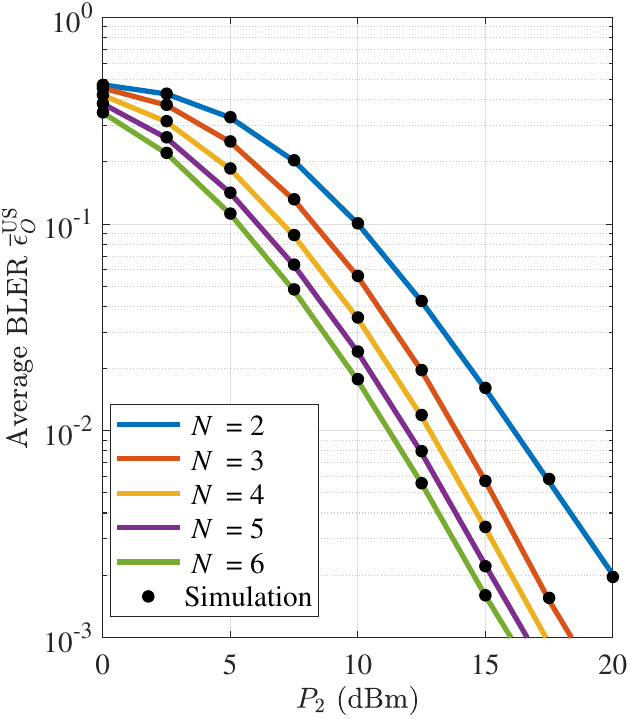}}
\caption{End-to-end BLER versus UAV transmit power $P_2$ for different numbers of FAS ports $N$. (a) Rural scenario: $L=100, m_1=5, m_2=5, P_1=20$ dBm. (b) Urban scenario: $L=100,m_{\text{LoS}}=5, m_{\text{NLoS}}=1, P_1=40$ dBm.}\label{fig:bler_vs_n}
\vspace{-3mm}
\end{figure}

Fig.~\ref{fig:bler_vs_n} provides a detailed investigation into the impact of the number of ports, $N$, on the BLER. Firstly, the results serve as a crucial verification of our theoretical derivations. The Monte Carlo simulation results are in excellent agreement with our analytical expressions, which confirms the accuracy of our analysis. As theoretically anticipated, increasing $N$ expands the port selection pool, substantially enhancing the spatial diversity gain by increasing the probability of finding a port with favorable channel conditions, thereby significantly lowering the BLER. A key observation is the parallel nature of the performance curves, particularly in the high-SNR regime, when plotted on a log-log scale. This parallelism strongly indicates that the achieved diversity order is directly proportional to $N$, confirming that each additional port actively and effectively contributes to the link's robustness.

\begin{figure}[t]
\centering
\subfloat[Rural scenario ]{%
\includegraphics[width=0.24\textwidth]{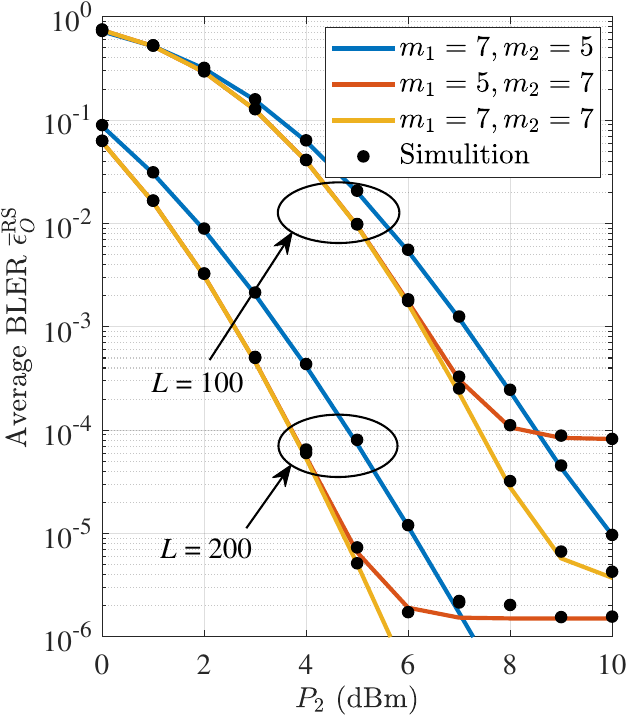}}
\hfil
\subfloat[Urban scenario ]{%
\includegraphics[width=0.24\textwidth]{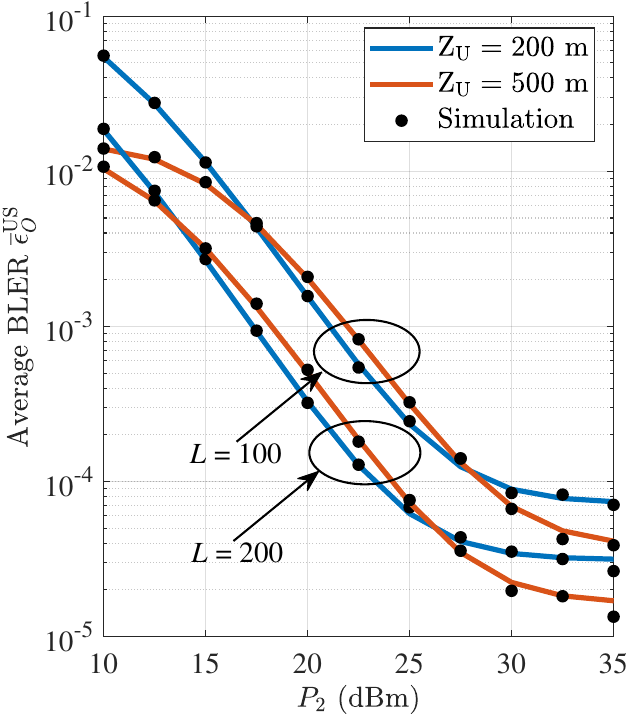}}
\caption{Impact of Nakagami-$m$ factors and blocklength $L$ in the rural scenario (a), and UAV altitude $Z_U$ and blocklength $L$ in the urban scenario (b).}\label{fig:param_impact}
\vspace{-3mm}
\end{figure}

% \subsection{Impact of Channel Parameters and Blocklength}
Fig.~\ref{fig:param_impact} examines the influence of key channel and system parameters on BLER performance. Fig.~\ref{fig:param_impact}(a) illustrates the effect of fading severity ($m$-factors) and blocklength ($L$) in the rural scenario. As shown, increasing the $m$-factor of either hop (e.g., from $5$ to $7$) significantly improves performance, as a higher $m$-factor implies a less severe fading environment. Concurrently, increasing the blocklength from $L=100$ to $L=200$ provides a substantial coding gain, thus effectively reducing the BLER for any given $P_2$.

Fig.~\ref{fig:param_impact}(b) studies the impact of UAV altitude $Z_U$ and blocklength in the urban scenario. It reveals an intricate interplay: at low $P_2$, the system is blockage-limited, and the higher altitude ($Z_U=500$ m) performs better due to its higher LoS probability. But as $P_2$ increases, the system becomes path-loss-limited. With sufficient power to overcome NLoS conditions, the lower altitude ($Z_U=200$ m) becomes more advantageous due to its reduced free-space path loss, causing the performance curves to cross. This non-monotonic behavior highlights that no single altitude is optimal across all power regimes and motivates the joint altitude optimization.

\begin{figure}[t]
\centering
\subfloat[Rural scenario]{%
\includegraphics[width=0.24\textwidth]{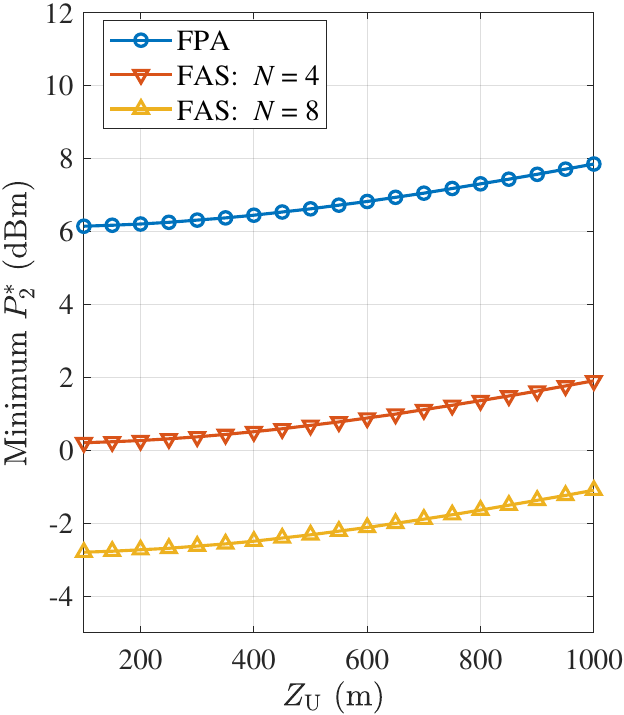}}
\hfil
\subfloat[Urban scenario]{%
\includegraphics[width=0.24\textwidth]{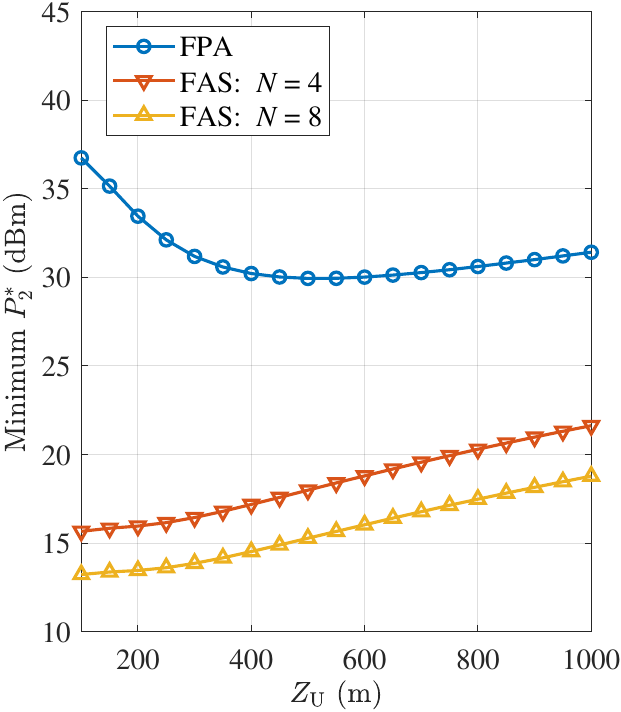}}
\caption{Minimum UAV transmit power $P_2^*$ versus UAV altitude $Z_U$ for $\epsilon_{\text{th}}=10^{-3}$ and $L=200$. Specific parameters: $P_1=15$~dBm (Rural) and $P_1=46$~dBm (Urban).}\label{fig:power_vs_alt}
\vspace{-3mm}
\end{figure}

%\subsection{Energy Efficiency Optimization}
Fig.~\ref{fig:power_vs_alt} investigates the minimum UAV transmit power $P_2^*$ required to meet $\epsilon_{\text{th}} = 10^{-3}$ at $L=200$, revealing two key findings. First, the proposed FAS ($N=4, 8$) substantially outperforms the conventional FPA ($N=1$) baseline in both scenarios, achieving over $15$ dB of power savings in the urban case. This shows the significant diversity gain harvested by FAS. Second, the optimal deployment strategy is scenario-dependent: the rural scenario (Fig. \ref{fig:power_vs_alt}(a)) is path-loss dominant, favoring the lowest altitude, whereas the urban scenario (Fig.~\ref{fig:power_vs_alt}(b)) exhibits a convex trend with an optimal altitude at $Z_U^* \approx 450$~m, balancing NLoS blockage and path loss. In both cases, the incremental gain from $N=4$ to $N=8$ is less than from $N=1$ to $N=4$, illustrating diminishing returns and validating the premise of our subsequent EE maximization.

\begin{figure}[t]
\centering
\subfloat[Rural scenario]{%
\includegraphics[width=0.24\textwidth]{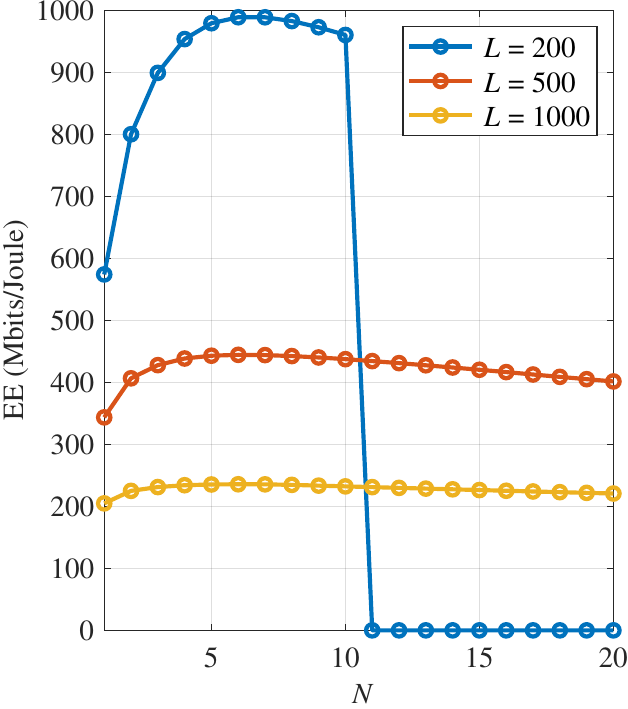}}
\hfil
\subfloat[Urban scenario]{%
\includegraphics[width=0.24\textwidth]{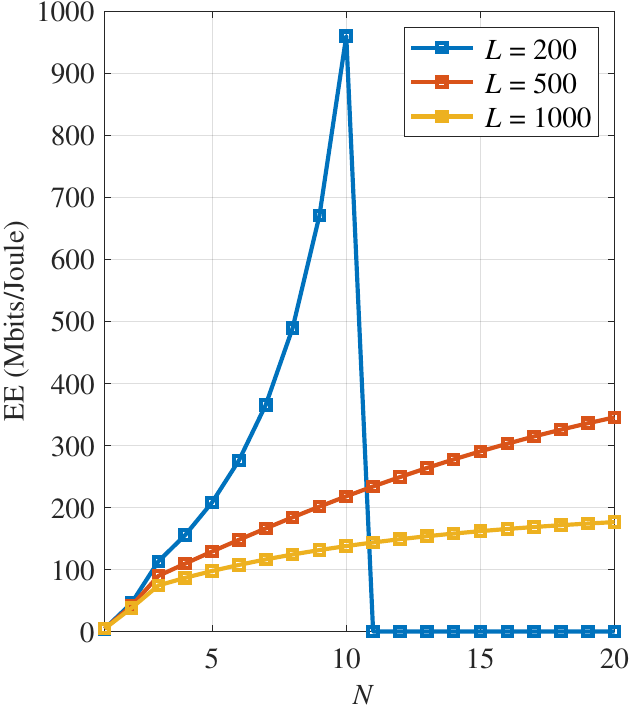}}
\caption{EE versus the number of FAS ports $N$ for different blocklengths $L$, with $\epsilon_{\text{th}} \le 10^{-3}$.}\label{fig:ee_vs_n}
\vspace{-3mm}
\end{figure}

Fig.~\ref{fig:ee_vs_n} presents the central finding of this paper. It plots the maximized EE as a function of $N$ for different blocklengths. The EE is clearly quasi-concave with respect to $N$. Initially, increasing $N$ enhances diversity, reducing the required $P_2^*$ and boosting EE. However, beyond an optimal $N^*$ (e.g., $N^* \approx 8$ for $L \ge 500$ in the rural scenario), the linear increase in time and energy overhead from port selection begins to outweigh the diminishing diversity gains, causing the EE to decrease. This demonstrates the critical importance of our realistic model. Furthermore, for the low-latency case of $L=200$, the EE abruptly drops to zero at $N=10$ due to the causality constraint $N\tau_p < L/W_{\text{band}}$, highlighting a hard physical limit on the number of ports in URLLC systems.

\begin{figure}[t]
\centering
\subfloat[Rural scenario]{%
\includegraphics[width=0.24\textwidth]{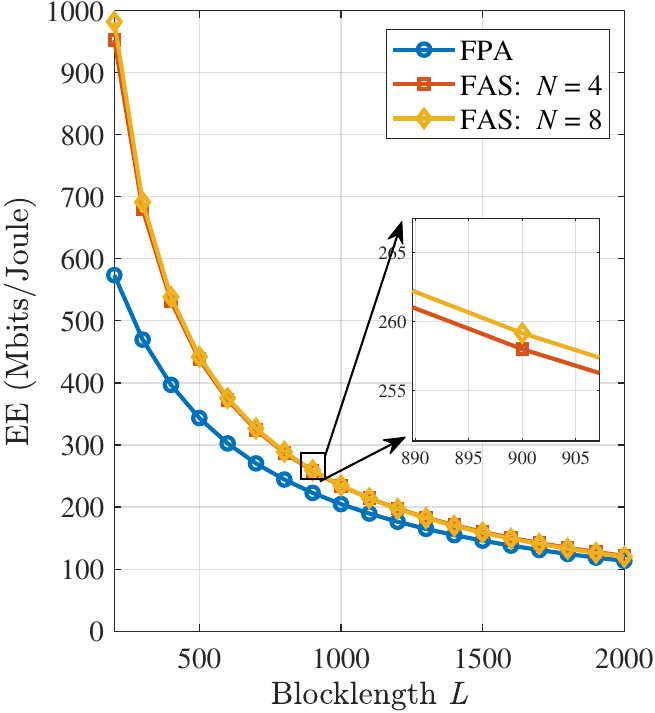}}
\hfil
\subfloat[Urban scenario]{%
\includegraphics[width=0.24\textwidth]{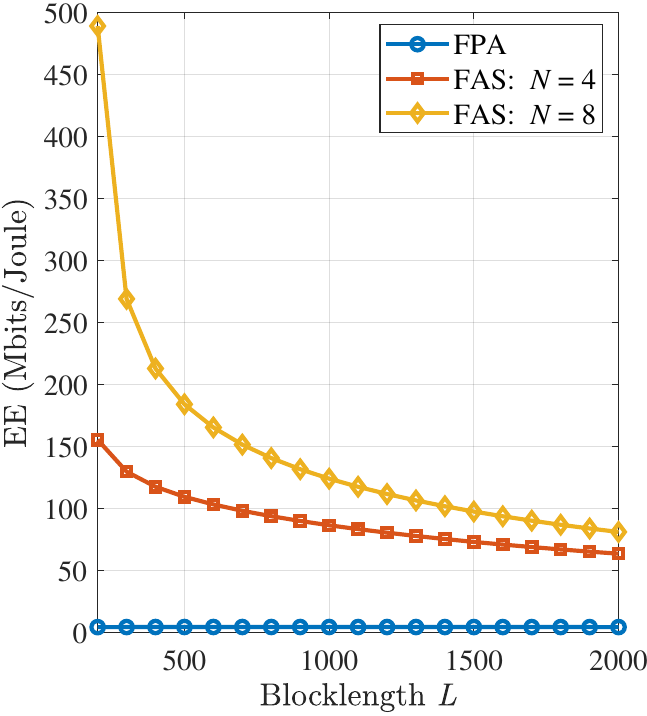}}
\caption{EE versus blocklength $L$ for different FAS configurations, with $\epsilon_{\text{th}} \le 10^{-3}$.}\label{fig:ee_vs_l}
\vspace{-3mm}
\end{figure}

Next, we use the results in Fig.~\ref{fig:ee_vs_l} to analyze the EE as a function of $L$. In both scenarios, FAS ($N>1$) consistently and substantially outperforms the FPA ($N=1$) baseline. Notably, in the urban scenario, the FPA system is barely viable, emphasizing that FAS is an enabling technology for efficient and reliable communications in such environments. The EE for all configurations generally decreases as $L$ increases. This is because the data payload $B$ is fixed, so a longer blocklength increases the total circuit energy consumption ($P_c \cdot T_{\text{block}}$) without increasing the transmitted bits, thus reducing efficiency.

\begin{figure}[t]
\centering
\subfloat[Rural scenario]{%
\includegraphics[width=0.24\textwidth]{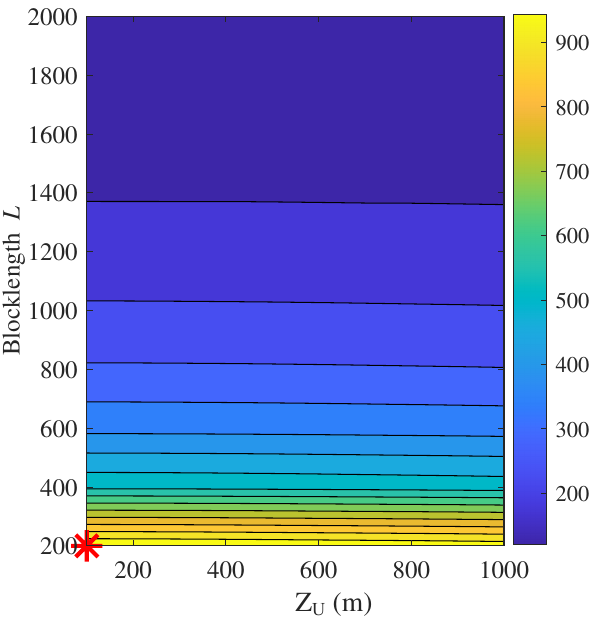}}
\hfil
\subfloat[Urban scenario ]{%
\includegraphics[width=0.24\textwidth]{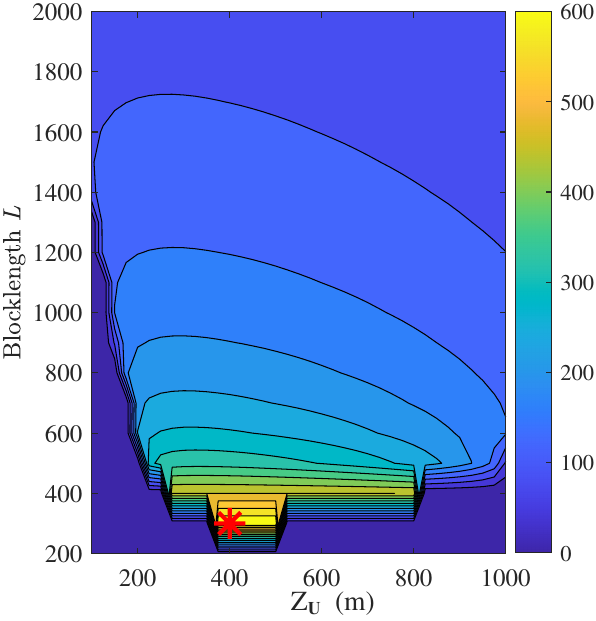}}
\caption{Contour map of the maximum achievable EE over the $(Z_U, L)$ plane, with $\epsilon_{\text{th}} \le 10^{-3}$. The red star marks the globally optimal operating point.}\label{fig:ee_contour}
\vspace{-3mm}
\end{figure}

Finally, Fig.~\ref{fig:ee_contour} provides a global view of the optimization landscape by plotting the maximum achievable EE over the $(Z_U, L)$ plane. At each coordinate, the EE value is the maximum obtained by optimizing over all $N$. The results visually confirm our findings. In the rural scenario (Fig. \ref{fig:ee_contour}(a)), the optimal region is at the lowest altitude ($Z_U^* = 100$~m) and shortest blocklength ($L^* = 200$). In contrast, the urban scenario (Fig.~\ref{fig:ee_contour}(b)) reveals a more complex landscape, with the optimal point found at a trade-off altitude of $Z_U^* \approx 400$~m and an intermediate blocklength ($L^* \approx 300$-$400$), validating the necessity of our holistic optimization framework.

\section{Conclusion}\label{sec:conclude}
In this paper, we presented a foundational framework for analyzing FAS-assisted UAV relay networks operating in the finite blocklength regime. We proposed a comprehensive analytical methodology to evaluate the end-to-end BLER, deriving novel closed-form expressions for both rural and urban environments. Our analysis, enabled by high-SNR asymptotic expressions, quantified the diversity order and identified the first-hop link as the ultimate performance bottleneck. Building upon this rigorous analytical foundation, we investigated the practical maximization of EE. By formulating a realistic model that explicitly accounts for the time and energy overhead associated with FAS port selection, we uncovered the central trade-off between diversity gain and operational cost. Our results demonstrated the existence of an optimal, finite number of ports and revealed that optimal UAV deployment strategies differ fundamentally between rural and urban scenarios. This research provides a unified theoretical and practical framework, offering valuable design guidelines for reliable and energy-efficient UAV communication systems. Future work could extend this framework to consider multiuser scenarios and study the impact of imperfect CSI.

\begin{appendices}
\section{Proof of Lemma \ref{lema1}}\label{plema1}
The PDF of a Nakagami-$m$ distributed channel amplitude $|g|$ is given by
\begin{equation}\label{nara1x}
f_{|g|}(x) = \frac{2m_1^{m_1}}{\Gamma({m_1})} x^{2m_1-1} e^{-{m_1} x^2},~x \geq 0
\end{equation}
where $\Gamma(\cdot)$ is the Gamma function and $m$ is the fading severity parameter.

Our goal is to find the CDF of the instantaneous channel power gain, $Y = |g|^2$. The CDF of $Y$ can be formulated as
\begin{equation}
F_Y(y) =\int_0^{\sqrt{y}} f_{|g|}(x) \, dx = \int_0^{\sqrt{y}} \frac{2{m_1}^{m_1}}{\Gamma({m_1})} x^{2{m_1}-1} e^{-{m_1} x^2} \, dx.
\end{equation}
Let us define $u = m_1 x^2$ and then we have
\begin{equation}\label{nara1x1}
F_Y(y) = \frac{1}{\Gamma({m_1})} \int_0^{{m_1}y} u^{{m_1}-1} e^{-u}  du.
\end{equation}
Because $\gamma(s,z) = \int_0^z t^{s-1}e^{-t}dt$, \eqref{nara1x1} can be written as
\begin{equation}
F_Y(y) = \frac{\gamma({m_1}, {m_1}y)}{\Gamma({m_1})}.
\end{equation}
For integer values of $m_1$, this expression can be simplified. Using the identity $\gamma(s,z) = (s-1)!(1 - e^{-z}\sum_{k=0}^{s-1} z^k/k!)$ and $\Gamma(s) = (s-1)!$, we obtain
\begin{align}
F_Y(y) &= 1 - \frac{\Gamma({m_1}, {m_1}y)}{\Gamma({m_1})},\label{eq:cdf_Y_gamma}\\
F_{\gamma_1^{\rm RS}}(\gamma)
% = &P(C Y \leq \gamma) = F_Y\left(\frac{\gamma}{C}\right)\\
% =& 1 - \frac{\Gamma({m_1}, {m_1} \gamma / C)}{\Gamma({m_1})} = 1 - \frac{\Gamma({m_1}, \gamma\vartheta_1)}{\Gamma({m_1})}\\
&=1 - e^{-\gamma\vartheta_1} \sum\nolimits_{k=0}^{{m_1}-1} \frac{(\gamma\vartheta_1)^k}{k!},\label{eq702}
\end{align}
where $\vartheta_{1}=m_1/\bar{\gamma}_1^{\mathrm{RS}}(\theta)$.

Since $\gamma_1^{\text{RS}} = \bar{\gamma}_1^{\text{RS}} Y$, the corresponding CDF is obtained by substituting $y = x / \bar{\gamma}_1^{\text{RS}}$ into \eqref{eq702}. Letting $\vartheta_1 = m_1/\bar{\gamma}_1^{\text{RS}}(\theta)$, we get the result \eqref{eq:cdf_snr_closed}, which completes the proof.

\section{Proof of Lemma \ref{lema2} }\label{plema2}
% The proof starts from the analytical model for the effective channel power gain, as defined in the main paper (e.g., equation (2)), i.e.,
% $|h_{\mathrm{FAS}}|^2 \triangleq \max\{\lambda_1|g_1|^2, \dots, \lambda_{N_{\mathrm{eff}}}|g_{N_{\mathrm{eff}}}|^2\} \triangleq Y_{\mathrm{max}}$.
Letting $Y_n \triangleq \lambda_n|g_n|^2$, the CDF of $Y_{\max}$ can be found as
\begin{align}
F_{|h_{\rm FAS}|^2}(y) &= \Pr(\max\{Y_1, \dots, Y_{N_{\rm eff}}\} \le y)\notag\\
&= \Pr(Y_1 \le y, Y_2 \le y, \dots, Y_{N_{\rm eff}} \le y)\notag\\
&= \prod_{n=1}^{N_{\rm eff}} \Pr(Y_n \le y) = \prod_{n=1}^{N_{\rm eff}} F_{Y_n}(y)\notag\\
&= \prod_{n=1}^{N_{\rm eff}} \frac{\gamma\left(m_2, \frac{m_2y}{\lambda_n}\right)}{\Gamma(m_2)}.\label{eq:cdf_ymax_expanded}
\end{align}
Recalling that $\gamma_2^{\rm RS} = \frac{P_2 \beta_2^{\rm RS}(\theta)}{\sigma^2} |h_{\rm FAS}|^2$,
the CDF of $\gamma_2^{\rm RS}$ is found by scaling $F_{|h_{\rm FAS}|^2}(y)$ as
\begin{equation}
\begin{aligned}
F_{\gamma_2^{\rm RS}}(x) &= \Pr(\gamma_2^{\rm RS} \le x)
= F_{|h_{\rm FAS}|^2}\left( x \left[ \frac{\sigma^2}{P_2 \beta_2^{\rm RS}(\theta)} \right] \right).
\end{aligned}\label{eq:snr_scaling_corrected} %%% a?(R)ae-L: e?ae(75)cae-Lc!(R)a1/2ca1/4
\end{equation}
Substituting \eqref{eq:snr_scaling_corrected} into \eqref{eq:cdf_ymax_expanded} and replacing $y$ with $\frac{x \sigma^2}{P_2 \beta_2^{\rm RS}(\theta)}$ gives
\begin{align}
F_{\gamma_2^{\rm RS}}(x)
&= \frac{1}{[\Gamma(m_2)]^{N_{\rm eff}}} \prod_{n=1}^{N_{\rm eff}} \gamma\left(m_2, \frac{m_2}{\lambda_n} \left[ \frac{x \sigma^2}{P_2 \beta_2^{\rm RS}(\theta)} \right] \right). \label{eq:final_cdf_intermediate}
\end{align}
We now use the definition of $\vartheta_2(\theta)$ from Lemma \ref{lema2}, which has $\vartheta_2(\theta) = m_2 \sum_{n=1}^{N_{\rm eff}}\lambda_n / \overline{\gamma}_2^{\rm RS}(\theta)$. By substituting the definition of $\overline{\gamma}_2^{\rm RS}(\theta)$ from (\ref{eq:gamma2}), $\vartheta_2(\theta)$ simplifies to
\begin{equation}
    \vartheta_2(\theta) = \frac{m_2 \sum_{n=1}^{N_{\rm eff}}\lambda_n}{ \left( P_2 \beta_2^{\rm RS}(\theta) \sum_{n=1}^{N_{\rm eff}}\lambda_n \right) / \sigma^2 } = \frac{m_2 \sigma^2}{P_2 \beta_2^{\rm RS}(\theta)}. \label{eq:theta_derivation}
\end{equation}
By substituting the result from \eqref{eq:theta_derivation} into \eqref{eq:final_cdf_intermediate}, we have
\begin{align}\label{eq:cdf_gamma2_final_corrected}
F_{\gamma_2^{\rm RS}}(x) &= \frac{1}{[\Gamma(m_2)]^{N_{\rm eff}}} \prod_{n=1}^{N_{\rm eff}} \gamma\left(m_2, \frac{x \vartheta_2(\theta)}{\lambda_n}\right).
\end{align}
When $m_2$ is a positive integer, we use the identity $\gamma(m_2, z) = (m_2-1)!(1-e^{-z}\sum_{k=0}^{m_2-1}\frac{z^k}{k!})$ and $\Gamma(m_2)=(m_2-1)!$. Substituting these into \eqref{eq:cdf_gamma2_final_corrected} yields the final form of the CDF as presented in Lemma \ref{lema2}. This completes the proof.

\section{Proof of Theorem \ref{theoreapd1}}\label{prooftheoreapd1}
The high SNR regime corresponds to small values of the argument $z$ in the CDF expression. For $z \to 0$, the lower incomplete gamma function $\gamma(s, z)$ can be approximated by the first term of its series expansion as
\begin{equation}
\gamma(s, z) = \sum_{k=0}^\infty \frac{(-1)^k}{k!} \frac{z^{s+k}}{s+k} \approx \frac{z^s}{s}, ~ \text{for } z \to 0.
\end{equation}
Recall that the CDF from Lemma \ref{lema2} is expressed as $F_{\gamma_2^{\rm RS}}(x) = \frac{1}{[\Gamma(m_2)]^{N_{\rm eff}}} \prod_{n=1}^{N_{\rm eff}} \gamma\left(m_2, \frac{x \vartheta_2}{\lambda_n}\right)$. By using the approximation $\gamma(m_2, z) \approx z^{m_2}/{m_2}$ to each term in the product, we get
\begin{equation}
F_{\gamma_2^{\rm RS}}(x) \approx
\frac{1}{[m_2\Gamma(m_2)]^{N_{\rm eff}}} \prod_{n=1}^{N_{\rm eff}} \left(\frac{\vartheta_2^{m_2} x^{m_2}}{\lambda_n^{m_2}}\right).
\end{equation}
Separating the terms that depend on $n$, we obtain
\begin{equation}
F_{\gamma_2^{\rm RS}}(x) \approx
\frac{(\vartheta_2^{m_2})^{N_{\rm eff}}}{[m_2\Gamma(m_2)]^{N_{\rm eff}}} \left(\prod_{n=1}^{N_{\rm eff}} \lambda_n^{-{m_2}}\right) x^{m_2{N_{\rm eff}}}.
\end{equation}
Using the identity $\Gamma(m_2+1) = m_2\Gamma(m_2)$, the expression simplifies to
\begin{equation}\label{eq:cdf_approx_final}
F_{\gamma_2^{\rm RS}}(x) \approx \left(\frac{\vartheta_2^{m_2}}{\Gamma(m_2+1)}\right)^{N_{\rm eff}}\left(\prod_{n=1}^{N_{\rm eff}} \lambda_n^{-m_2}\right) x^{m_2 {N_{\rm eff}}}.
\end{equation}
Define a constant $\mathcal{C}$ that groups all terms independent of $x$ as
\begin{equation}
\mathcal{C} \triangleq \left(\frac{\vartheta_2^{m_2}}{\Gamma(m_2+1)}\right)^{N_{\rm eff}} \left(\prod_{n=1}^{N_{\rm eff}} \lambda_n^{-m_2}\right).
\end{equation}
The approximated CDF is now simply $F_{\gamma_2^{\rm RS}}(x) \approx \mathcal{C} x^{m_2{N_{\rm eff}}}$. The average BLER is given by an integral involving the CDF. Based on the approximation, we have
\begin{equation}
\bar{\epsilon}^{\rm RS}_2 = \chi \int_{\rho_L}^{\rho_H} F_{\gamma_2^{\rm RS}}(x) dx \approx \chi \int_{\rho_L}^{\rho_H} \mathcal{C} x^{m_2 {N_{\rm eff}}} dx.
\end{equation}
After that, we perform the integration as
\begin{equation}
\bar{\epsilon}^{\rm RS}_2 \approx
\frac{\chi \mathcal{C}}{m_2 {N_{\rm eff}} + 1} \left( \rho_H^{m_2 {N_{\rm eff}} + 1} - \rho_L^{m_2 {N_{\rm eff}} + 1} \right).
\end{equation}
Substituting the full expression for $\mathcal{C}$ back gives the result.
\end{appendices}

\end{document}